%% file: Doppler_space_key_arxiv.tex
\documentclass{IEEEtran}
\IEEEoverridecommandlockouts
\usepackage{cite}
\usepackage{amsmath,amssymb,amsfonts}
\usepackage{amsthm}
\usepackage{balance}
\usepackage{algorithmic}
\usepackage{graphicx}
\usepackage{textcomp}
\usepackage{float}
\usepackage{xcolor}
\usepackage{multirow}
\usepackage{subfigure}
\usepackage{adjustbox}
\usepackage{stfloats}
\def\BibTeX{{\rm B\kern-.05em{\sc i\kern-.025em b}\kern-.08em T\kern-.1667em\lower.7ex\hbox{E}\kern-.125emX}}
\newtheorem{proposition}{Proposition}
\newtheorem{corollary}{Corollary}

\begin{document}

\title{\LARGE Securing the Inter-Spacecraft Links: \\ Physical Layer Key Generation from Doppler Frequency Shift \\
{}
%\thanks{This work has been supported by the Tubitak 115E827.}
}

\author{ \IEEEauthorblockN{ Ozan Alp Topal, \IEEEmembership{Graduate Student Member, IEEE}, Gunes Karabulut Kurt, \IEEEmembership{Senior Member, IEEE}, Halim Yanikomeroglu, \IEEEmembership{Fellow, IEEE}}

\thanks{{ O. A. Topal and G. Karabulut Kurt are with the {Department of Electronics and Communication Engineering}, {Istanbul Technical University}, Istanbul, 34469 Turkey, (email:	\{topalo, gkurt\}@itu.edu.tr).}

 G. Karabulut Kurt is also with the Department of Electrical Engineering,  Polytechnique Montr\'eal, Montr\'eal, Canada, (e-mail: gunes.kurt@polymtl.ca).

%G. K. Kurt was with the {Department of Electronics and Communication Engineering}, {Istanbul Technical University}, Istanbul, 34469 Turkey, when this work was performed.  She is now with the Department of Electrical Engineering Polytechnique Montr\'eal, Montr\'eal, QC, H3C3A7, Canada (email: gunes.kurt@polymtl.ca).
		 
{H. Yanikomeroglu is with the Department of Systems and Computer Engineering}, {Carleton University}, Ottawa, ON, K1S 5B6, Canada,	(email: halim@sce.carleton.ca).}

%\and
%\IEEEauthorblockN{Gunes Karabulut Kurt}
%\IEEEauthorblockA{\textit{Department of Electronics and Communication} \\
%\textit{Istanbul Technical University}\\
%Istanbul, Turkey \\
%gkurt@itu.edu.tr}
%\and
%\IEEEauthorblockN{Halim Yanikomeroglu}
%\and
%\and
%\and
}

\maketitle

\begin{abstract}
In this work, we propose a secret key generation procedure specifically designed for the inter-spacecraft communication links. As a novel secrecy source, the spacecrafts utilize Doppler frequency shift based measurements. In this way, the mobilities of the communication devices are exploited to generate secret keys, where this resource can be utilized in the environments that the channel fading based key generation methods are not available. The mobility of a spacecraft  is modeled as the superposition of a pre-determined component and a dynamic component. We derive the maximum achievable secret key generation rate from the Doppler frequency shift. The proposed secret key generation procedure extracts the Doppler frequency shift in the form of nominal power spectral density samples (NPSDS). We propose a maximum-likelihood (ML) estimation for the NPSDS at the spacecrafts, then a uniform quantizer is utilized to obtain secret key bits. The key disagreement rate (KDR) is analytically obtained for the proposed key generation procedure. Through numerical studies, the tightness of the provided approximations {\color{black} is} shown. Both the theoretical and numerical results demonstrate the validity and the practicality of the presented physical layer key generation procedure considering the security of the communication links of spacecrafts. 
\end{abstract}


\section{Introduction}
Exploring, understanding, and colonizing other planets {\color{black} become} the future roadmap for human existence in the space. The availability of low-cost devices and advanced rocket technology make this roadmap more attainable, where developing robust communication systems {\color{black}is an} integral part for future space missions. As a starting point, today mostly the small-space satellites are launched into low-orbital space by top-tier companies like SpaceX, Google, Facebook, Virgin Galactic\cite{first_ref}. 
 
 However, exploration of outer space is expected to become the near future competition for countries and private companies. For example, with the Artemis program, NASA targets to land the first human on the Moon by 2024\cite{NASA}. By the end of the decade, NASA is planning to form sustainable space missions that eventually sending astronauts to Mars. Another important space mission is carried out by Chinese National Space Administration, where a relay satellite is utilized to establish a reliable communication link in the dark side of the Moon.  In order to join the space exploration race, these spacecrafts require innovative technologies that provide a lower cost of production, launch, and maintenance. This requirement limits the sensing and the communication capabilities predicting their limited size, weight, and power. Considering these characteristics, the spacecrafts that will be utilized in the future space missions act as space cyber-physical networks, where the control and connectivity of many low-cost and software-enabled controllable devices are the main priority. The security of the communication channel and communication devices becomes the priority of the cyber-physical systems in order to secure the whole network from the large scale accidents initiated by the security breaches of low-cost devices \cite{magazin}.

 %By providing continuous global and regional coverage, the space networks promise to support the \textit{connected everytime-everywhere} vision of the 6G networks. Especially, the space networks may enable wireless Internet connectivity for the unconnected population in the rural regions. As reported in \cite{6G_Rural}, 75\% of the unconnected population that reside in 20 countries are mostly concentrated in rural areas, where the extreme cost of the terrestrial infrastructure makes the connectivity unattainable. Developing an integrated satellite and terrestrial network architecture is critically important for industries such as logistics, agriculture, and defense. In order to obtain the global connectivity vision, satellite constellations which consist of thousands of satellites  are needed to be established. 
 
 \begin{figure}[t]
 	\centering
 	\includegraphics[width=\linewidth]{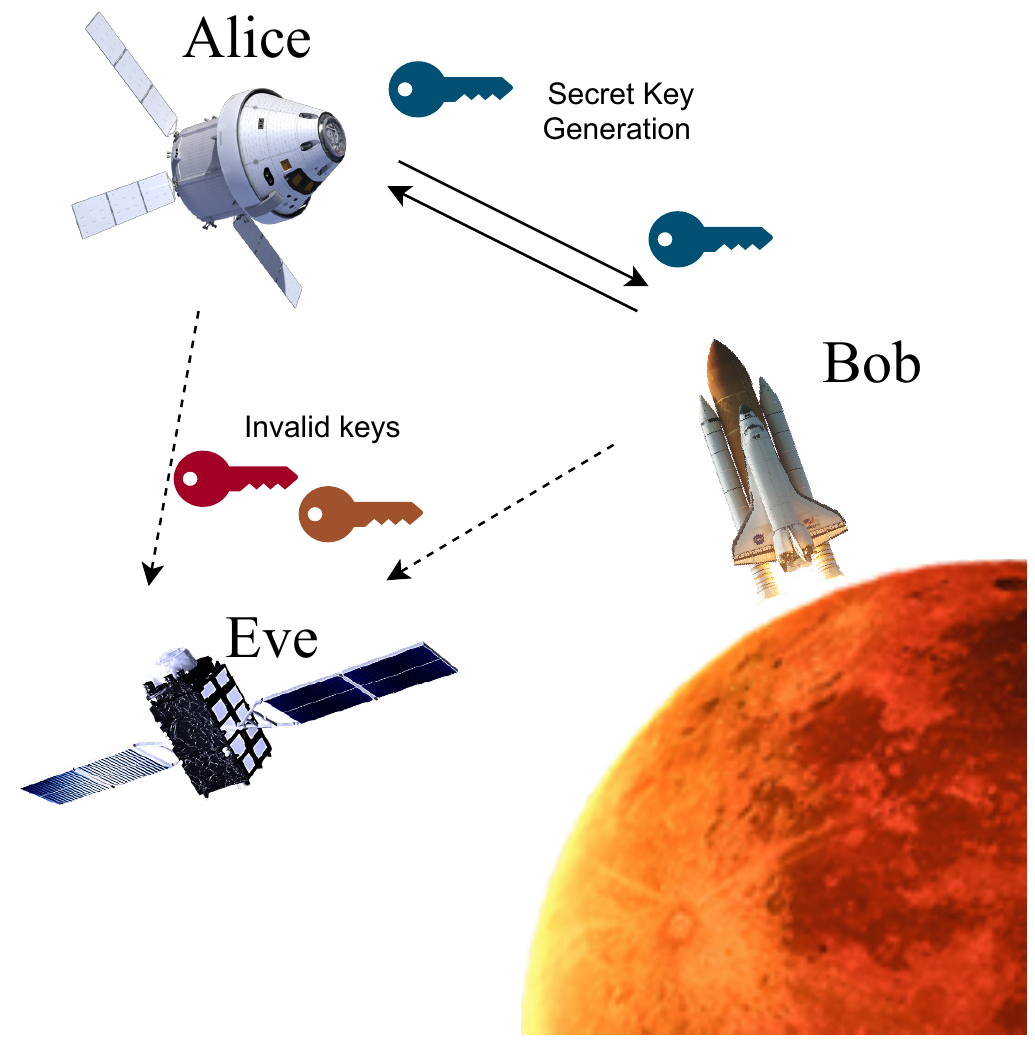}
 	\caption{An illustration of an inter-spacecraft link (ISL). Alice, Bob and Eve represent three different {\color{black}spacecraft}s. Alice and Bob try to establish a secure communication line, while Eve tries to eavesdrop their communication.}
 	\label{fig:LEO}
 \end{figure}

 %LEO satellites are roughly under 50 kilograms, and they {\color{black} require low-complexity}, where one of the biggest challenge of these constellations is continuous connectivity between one another or the continuous connectivity of the space network to the ground station.  \cite{IoT_space}. 
 
The open nature of the wireless communication channel creates security breaches for the space networks similar to the wireless cyber-physical systems. Especially considering the critical research and exploration information harnessed by the space mission, the security breach of an inter-spacecraft link (ISL) would create the loss of highly profitable information. Considering the high operational costs and technical difficulties of restarting a space mission, the security becomes a non-negotiable concept in space networks. Conventionally, in symmetric key encryption mechanisms, the security has enabled by cryptographic methods, where the message is encrypted with a bit sequence at the transmitter and decrypted with the bit sequence at the receiver. This sequence is named as the secret key, and the sharing of the secret key becomes a bottleneck for the low-battery power and low chip size communication devices. Especially to obtain a post-quantum secure communication systems, the longer secret keys are required to be shared, where the complexity of these operations  become  unattainable for the low-cost devices.

Alternatively, physical layer key generation methods are proposed for obtaining identical secret keys at the distant nodes without any key sharing complexities. The distant communication devices {\color{black}exploit} the unique characteristics of the physical layer link that only identical or very similar for them, and seem random for other observers to extract identical secret keys. For example, the received signal strengths have been proposed as a secrecy resource \cite{RSS} or the reciprocal channel estimates are utilized for physical layer secret key generation \cite{Topal}. Also, the spatial correlations of the channel fadings are analyzed as another resource of secrecy \cite{Topal_MIMO}. These methods provide {\color{black} a} high amount of secrecy under the assumption of {\color{black} a} dynamic environment and channel fading. However, under the static channel fading assumptions as in the deep space networks, these methods become inapplicable since the channel based attenuations become predictable. Therefore, the utilization of cryptical methods as well as channel based secret key generation methods becomes unfavorable for the inter-spacecraft communication networks.  

% Previously, received signal strengths (RSSs) \cite{RSS}, channel gains \cite{Topal}, channel phases or spatial correlations of the channel fadings \cite{Topal_MIMO} are utilized as a common secrecy source in physical layer key generation. Considering the ISLs, continuous randomness cannot be attained by both RSS and  channel fading based methods since the channel attenuation is resulted from path loss that is predictable. Based on the following observation, we propose a fresh perspective for the physical layer key generation in ISLs.
  
\subsection{Contributions}
As a solution to the discussed problem, in \cite{WiSee}, we have proposed a mobility based secret key generation scheme, where the spacecrafts exploits the reciprocal Doppler frequency shift measurements as a secrecy resource. More specifically, spacecrafts utilize the measurements of nominal power spectral density samples (NPSDS), where the effect of Doppler frequency {\color{black}shifts} are fully reserved. As a first step, the communication nodes exchange pilot signals, and estimate NPSDS with maximum likelihood (ML) estimator. The estimations are quantized to generate secret key bits. 

As an extension of our previous work in \cite{WiSee}, the main contributions of this paper can be enlisted as
% Motivated by this observation, we will present a novel Doppler frequency based secret key generation procedure. The proposed method is based on collecting pilot signals transmitted from two separated nodes and utilizing the nominal power spectral density samples (NPSDSs) that will be identical for symmetric Doppler frequencies.  As described in \cite{Spectrum}, one of the key requirements for a spacecraft is estimating and overcoming the Doppler frequency shift. We exploit this process to generate a common secret between two distant spacecrafts without introducing any additional complexity to the system. The main contributions of this paper can be listed as 
\begin{itemize}
\item We model the mobility of the spacecrafts as the superposition of a publicly known, predictable component and a dynamic and randomly modeled component.
\item The maximum achievable secret key rate from the Doppler frequency shift measurements are provided. The dynamic mobility model of the spacecrafts {\color{black}is} modeled with a normalized version of the Brownian motion to obtain the maximum amount of secrecy. 
\item The key disagreement rate (KDR) has been analytically modeled considering the given random mobility model.
\item The numerical analysis {\color{black}shows} the tightness of the given approximations, and illustrates the maximum achievable secret key rate and KDR performance.  
\end{itemize} 

\subsection{Organization}
In the following, we provide related work in physical layer key generation. In Section II, we provide the considered system model, and the performance of the Doppler frequency shift as a secrecy resource. In Section III,  we explain the proposed secret key generation procedure. In Section IV, we provide theoretical analysis on the key disagreement rate of the proposed method. In Section V, we provide numerical analysis. In Section VI, we conclude the paper and provide the future work directions.
\subsection{Notation}
Scalar variables are denoted by italic symbols, vectors are denoted by boldface symbols. $\rho(x)$ denotes the probability density function of the random variable $X$. $\mathcal{N}(\mu,\sigma^2)$ denotes the normal random variable with mean $\mu$ and variance $\sigma^2$. The non-central chi-squared distribution with $k$ degrees of freedom and $\delta$ non-centrality parameter is denoted by ${\chi'}_k^2(\delta)$. $\text{erf}(\cdot)$ denotes the error function as described in \cite{elements}. $h(X)$ denotes the differential entropy of the random variable $X$. $\text{log}(\cdot)$ denotes the natural logarithm. $\mathcal{I}_v(\cdot)$ denotes the $v^{\text{th}}$ order modified Bessel function. $||\mathbf{d}||$ denotes the Euclidean norm of the vector $\mathbf{d}$. $\mathbb{E}\{\cdot\}$ denotes the expectation operator. 

%{\color{black}
%	Problems:
%	\begin{itemize}
%		\item We consider RF bandwidth but microwave, THz and optical comm is more popular on intersatellites 
%		\item KDR derivation normalization  
%		\item nominal power specral sample too long / NPSS ?
%		\item Doppler vurgusu kuvvetli ancak NPSS kullanıyoruz bu gecisi baska bolumlerde vurgulamalı mı
%		\item Numerical results'da pdf ve MSE resultları benzer seyler anlatıyor 
%\end{itemize}}
\section{Related Works}
\begin{table*}[t]
	\centering
	\footnotesize
	\caption{Comparison of the related work.}
	\color{black}
	\begin{tabular}{|l|l|l|l|l|}
		\hline
		\textbf{Randomness Resource}               & \textbf{Reference}                    & \textbf{Implementation} & \textbf{Channel assumption} & \textbf{\begin{tabular}[c]{@{}l@{}}Applicability\\ to Space Networks\end{tabular}} \\ \hline
		\multirow{3}{*}{Received signal strengths} & \cite{rss_1}        & indoor                  & Non Line Of Sight (NLOS)    & No                                                                                 \\ \cline{2-5} 
		& \cite{rss_2}        & outdoor                 & NLOS  and deep fade         & No                                                                                 \\ \cline{2-5} 
		& \cite{rss_3}, \cite{csi_vs_rss}        & indoor                  & LoS and NLoS                & No                                                                                 \\ \hline
		\multirow{6}{*}{Channel amplitude}         &    \cite{Topal}                                   &     indoor                   &      LoS                       & No                                                                                 \\ \cline{2-5} 
		&         \cite{Topal_MIMO}                              &      indoor/outdoor                   &   NLoS, spatial correlation                          & No                                                                                 \\ \cline{2-5} 
		& \cite{beam_spatial} &          outdoor               &      LoS/ spatial correlation                       & Limited                                                                               \\ \cline{2-5} 
		&       \cite{csi_v2x}                                &     outdoor/vehicular                    &         LoS/NLoS - dynamic                    & No                                                                                \\ \cline{2-5} 
		&     \cite{csi_vs_rss}, \cite{csi_2}                                  &  indoor                       &     NLoS                        &              No                                                                     \\ \hline
		\multirow{1}{*}{Channel phase}             &  \cite{phase_1}                                      &       outdoor                &    NLoS                         & Yes                                                                      \\ \hline
		Power spectral density                     &        \cite{underwater}                              &               underwater         &       NLoS - dynamic                        & Limited                                                                            \\ \hline
		Voltage/current deviations                 &  \cite{power_line}                                     & power line              & n/a                         & No                                                                                 \\ \hline
		Carrier frequency offset                    &  \cite{cfo_1}                    &    outdoor/vehicular                     &       NLoS/LoS - dynamic                      & Yes                                                                                \\ \hline
		Doppler frequency shift                    & this work                    &   outdoor/ high-mobility                      &    LoS - highly dynamic                         & Yes                                                                                \\ \hline
	\end{tabular}
\label{tab:references}
\end{table*}
{\color{black} Table \ref{tab:references} provides a comparison for the state-of-the-art (SOTA) PHY key generation methods including their applicability for the space communication networks.  In terrestrial networks, especially for the sub 6 GHz systems, the multi-path fading is the most dominant factor that deteriorates the communication performance. The multipath fading in terrestrial networks is generally modeled by Rayleigh or Rician distributions. To ensure reliable communications, the devices already apply channel estimation and equalization. The wireless devices have the measurements of channel state information (CSI) or received signal strengths (RSS) available. One important enabler for utilizing these estimates for the encryption purposes is that they are reciprocal for two distant nodes within the coherence time \cite{Topal}. Therefore, they have identical secrets without  sharing any information about the secrets itself. The main limiting factor in the space networks is that the fading channel is dominantly affected by the path loss term, where the main path is very dominant to other reflections and diffractions. In most of the cases such as communication between two space shuttles, the EM waves propagate in space without encountering any blockage. In that case, only fading factor is the free-space path loss, where it is same for the devices with the same distances. Hence, CSI or RSS based secret key generation methods cannot be applied considering the space communication networks.

On the other hand, very high mobility of the spacecrafts result into frequency and phase shifts for communication devices. The channel phase-based secret key generation methods are already available in the literature that also given in Table \ref{tab:references}. They may also utilized in the space communication networks, but the estimation process of the channel phases are more complex, and may include more erroneous results than the rest of the secrecy extraction schemes \cite{phase}.

As the main novelty, in this work, we propose a Doppler frequency shift based secret key generation procedure first time in the literature.} As described in the Figure \ref{fig:LEO}, we consider any three ISLs of three distinct spacecrafts. In the considered scenario, we  Alice and Bob exploit the symmetric measurements of the Doppler frequencies as a secrecy source for the first time in the literature. Since the mobility of the spacecrafts are updated in each communication block, the randomness in the mobility reflects on the Doppler frequency measurements identically for the reciprocal links.

%In \cite{hanzo}, the authors provide fundamental steps of the physical layer key generation, and compare the existing channel-based key generation methods. Alternatively, the authors of \cite{phase} introduce a key generation mechanism based on channel phase under narrowband channel fading assumption. In \cite{static}, the authors utilize a relay node in secret key generation to overcome the static environment characteristics. In \cite{underwater}, the authors introduce a secret key generation system for underwater communication channels. The authors of \cite{group_mesh} utilize the distributed mesh antenna devices to generate a common secret key. In \cite{power_line}, the authors propose a novel secrecy source for the power line communications. The authors of \cite{beam_spatial} propose a secret key generation procedure from the spatial–temporal beam features for te millimeter-wave wireless networks. 

In addition to different key generation schemes, another major contribution is increasing the efficiency of the proposed key generation mechanisms. The authors of \cite{PCA} make use of principle component analysis to decrease the key disagreement rates at the nodes. In \cite{Topal_wavelet}, the authors utilize a wavelet-based pre-processing to eliminate the dissimilarities of the channel observations.  Besides pre-processing after obtaning raw secret keys, information reconciliation has been applied to eliminate the discrepancies in the generated keys. In \cite{reccompare}, the authors provide a general overview on the reconciliation performance of the error correcting codes. More recently, the authors of \cite{nature_blind} provide a more efficient key reconciliation protocol. In \cite{kurt2020polynomial}, the authors propose a polynomial interpolation based information reconciliation scheme. All of the reconciliation schemes may utilized in different secret key generation schemes including our work to eliminate remaining erroneous bits.

\section{System Model}
In this paper, we consider point to point communication link between two spacecrafts. As envisioned in the future space missions, these spacecrafts might be satellites orbiting around Earth, Moon or Mars, or space shuttles carrying {\color{black}humans} or goods. The spacecrafts for these missions are expected to have very high mobility that effects the communication links severely.{ \color{black} The spacecrafts and Doppler frequency shift utilization is given in Figure \ref{fig:spacecraft}.} Alice and Bob are two distant spacecrafts that {\color{black} want} to communicate with secrecy. In the meantime, we assume a third spacecraft, Eve, eavesdrops the communication between Alice and Bob to obtain the same secret of Alice and Bob. We assume that all spacecrafts have single omni-directional antenna for message transmission and reception. Alice and Bob are assumed to communicate in time division duplexing (TDD) fashion. Alice and Bob would require identical secret keys to encrypt/decrypt their communication messages, where Eve should not be able to obtain this secret key. Due to the high computational requirements, commonly utilized secret sharing mechanisms such as Diffie-Hellman is assumed to be not available in the considered scenario \cite{post_quantum}. {\color{black} As depicted in Figure \ref{fig:Doppler_utilization}, estimating the Doppler frequency shift is also an integral part for enabling reliable communication and navigation services in space \cite{NASA_dopp,NASA_nav}.  The spacecrafts already utilize Doppler frequency shifts to compansate the carrier frequency and navigate the spacecraft. Since pilot sharing and Doppler frequency measurements are already an integral part of the system, the additional cost of complexity is introduced by the quantization operation. This complexity is quite low in comparison to the many operations utilized by the transceivers such as channel estimation and equalization \cite{hanzo}.}  At least one of the legitimate nodes is assumed to have high mobility. In the following, we detail the mobility model for the spacecrafts, and provide the fundamental resource of the secrecy for the communication between Alice and Bob.

\begin{figure}[t]%!b
\centering
			\subfigure[]{
				\label{fig:doppler_system}			\includegraphics[width=\linewidth]{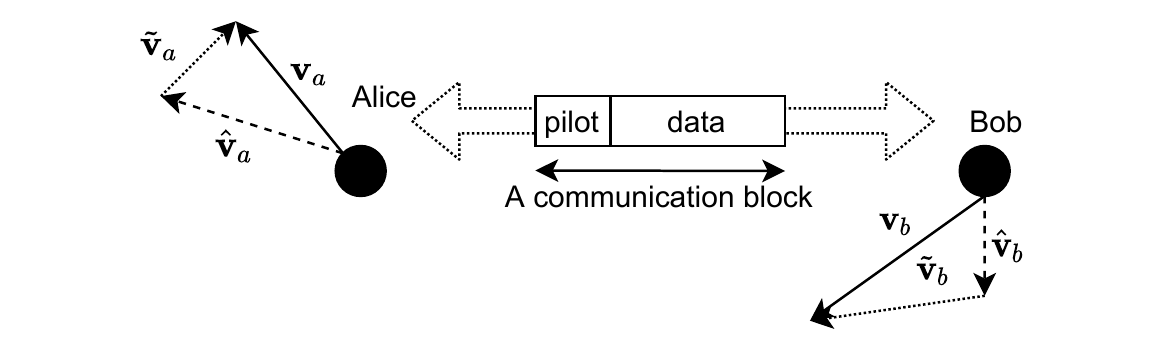} } \\
	
			\subfigure[]{
				\label{fig:Doppler_utilization}
			\includegraphics[width=\linewidth]{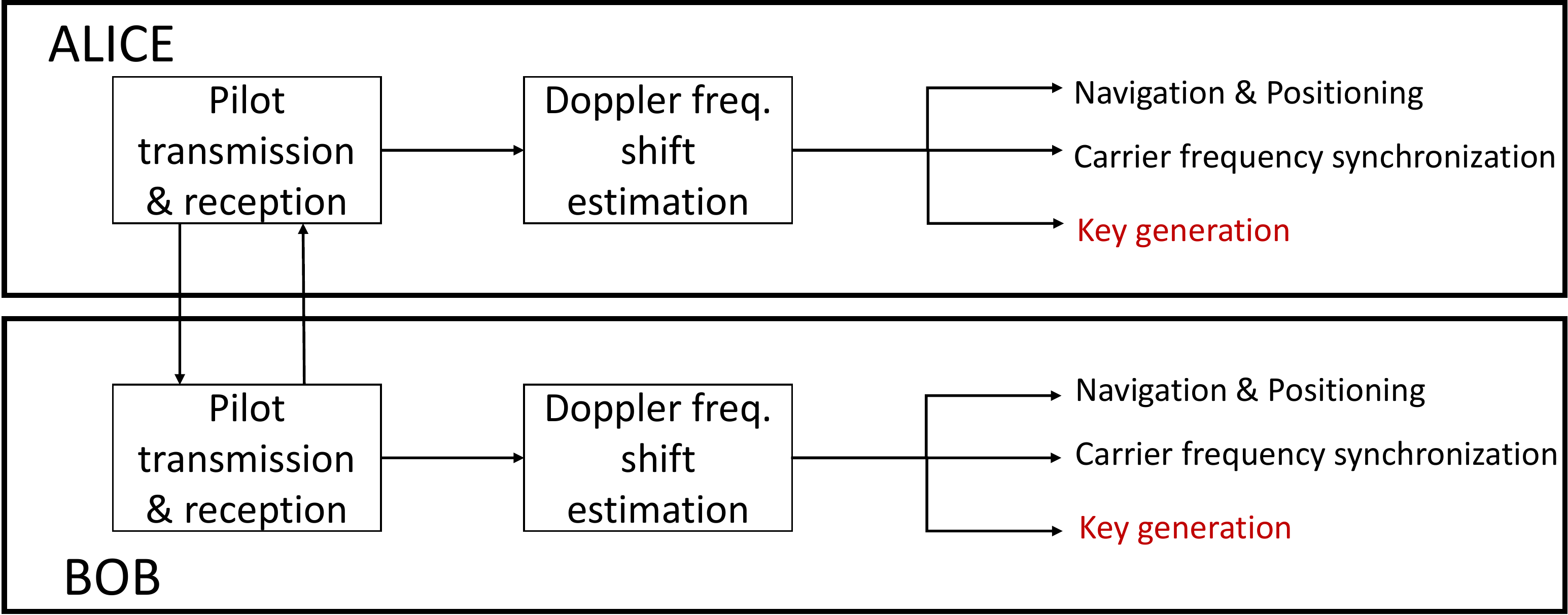} }
	\caption{\color{black}(a) A visiual representation for the relative velocity and position vectors, (b) Utilization of Doppler frequency shifts in spacecrafts.}
	\label{fig:spacecraft}
\end{figure}

\subsection{Spacecraft Mobility Model}
In the considered system model, the mobility of Alice and Bob are assumed to be composed of a static, expected movement component and a random, unexpected movement component. A communication block is defined for the time instant $t$, where Alice and Bob respectively share their messages respectively for each time instant. The velocities and the positions of the spacecrafts are updated for each time instant. 

Let us denote the velocity vector of the $k^{\text{th}}$ spacecraft at time instant $t$ in Cartesian coordinates by $\mathbf{v}_k(t)=[v^x_k,v^y_k,v^z_k]^T$, where $k\in\{a,b,e\}$ for respectively Alice, Bob and Eve. Similarly, the position vector of the $k^{\text{th}}$ spacecraft at time instant $t$ in Cartesian coordinates is denoted by $\mathbf{p}_k(t)=[p^x_k,p^y_k,p^z_k]^T$. In line with the mission, the spacecrafts are assumed to have a predetermined route and velocity values that known from all other spacecrafts. In the meantime, due to the dynamic environmental changes, the spacecrafts have small divergences from the aimed route and velocity values. We model the total amount of displacement and velocity as the superposition of these two effects: 
\begin{equation}
\mathbf{v}_k(t)= \mathbf{\hat{v}}_k(t)+\mathbf{\tilde{v}}_k(t),
\end{equation}
where $\mathbf{\hat{v}}_k(t)$ denotes the velocity of the main route that publicly known,  $\mathbf{\tilde{v}}_k(t)$ denotes the velocity introduced by the displacement due to the dynamic environmental changes. For example, for a LEO satellite orbiting {with a \color{black} static route, we can assume $\mathbf{\tilde{v}}_k(t)=0$.} However, for a space shuttle the dynamic environmental would cause higher divergences that result temporal discrepancies in the main route. {\color{black} Each spacecraft drafts from its expected trajectory due to numerous unpredictable reasons \cite{NASA_nav}. For example, the pressure of sunlight may catapult to a drift from the main trajectory. After drifting, the spacecraft starts to follow an arbitrary and random trajectory that the mobility controllers need to estimate and maneuver the spacecraft back to its main route. This random displacement is modeled in the manuscript by the $\tilde{v}_k$. Due to the unpredictability and the real-world randomness, even the space agencies consider these displacements as random behaviors. In this paper, we model  these  discrepancies in the velocity with a stochastic process, where the components in each Cartesian coordinate is \textit{i.i.d}, $\mathbf{\tilde{v}}_k(t)=[{\tilde{v}}^x_k,{\tilde{{v}}}^y_k,{\tilde{v}}^z_k]^T$.}

As discussed in above, due to the static channel environment in deep space communication, the legitimate parties may not utilize the channel fading randomness based secrecy sources to obtain an identical secret keys. The following subsection light the way for a novel secrecy source for distant communication parties in this particular challenging environment.

\subsection{Doppler Frequency Shift as a Secrecy Source}

Let $f_{\textcolor{black}{m}k}$ denote the nominal Doppler frequency, where $f_{\textcolor{black}{m}k}=\frac{c \omega_{\textcolor{black}{m}k}}{f_c}$, where $\omega_{\textcolor{black}{m}k}$ is the Doppler frequency observed at the $k^{\text{th}}$ spacecraft of the transmitted signal from the $\textcolor{black}{m}^{\text{th}}$ spacecraft, $c$ is the speed of light in m/s, and $f_c$ is the carrier frequency of the transmitted signal.  The nominal Doppler frequency at time instant $t$ is given by
\begin{equation}
\begin{aligned}
f_{k\textcolor{black}{m}}(t)&=\frac{\mathbf{v^T_{\textcolor{black}{m}}}(t)(\mathbf{p_{\textcolor{black}{m}}}(t)-\mathbf{p_{k}}(t))}{||\mathbf{p_{\textcolor{black}{m}}}(t)-\mathbf{p_{k}}(t)||}+\frac{\mathbf{v_{k}^T}(t)(\mathbf{p_{k}}(t)-\mathbf{p_{\textcolor{black}{m}}}(t))}{||\mathbf{p_{k}}(t)-\mathbf{p_{\textcolor{black}{m}}}(t)||}, \\ &=\frac{(\mathbf{v_{\textcolor{black}{m}}}(t)-\mathbf{v_{k}}(t))^T(\mathbf{p_{\textcolor{black}{m}}}(t)-\mathbf{p_{k}}(t))}{||\mathbf{p_{\textcolor{black}{m}}}(t)-\mathbf{p_k}(t)||}.
\end{aligned}
\end{equation}

The difference of velocity and position vectors provide us the relative velocity of the spacecrafts. The following observation is the cornerstone for the secrecy harnessing from the Doppler frequency shift.

\noindent\textbf{Observation 1.}\textit{ The Doppler frequency of the reciprocal links at a certain time instant can be uniquely described with $\omega_{k\textcolor{black}{m}}=-\omega_{\textcolor{black}{m}k}$, while the Doppler frequency of received signal from any other physically disjoint node $t$ will be different than other ISL's $\omega_{k\textcolor{black}{m}} \neq \omega_{k\textcolor{black}{n}} \neq \omega_{\textcolor{black}{m}\textcolor{black}{n}}$.} 

This observation guarantees that two distant spacecraft, Alice and Bob may obtain identical Doppler frequency observations resulted from their high mobilities. In the meantime, Eve also observe a Doppler frequency shift in her versions of the signal. As stated in the observation, since the relative velocity and position of Eve with Alice and Bob would be different than the relative speed of Alice and Bob, her Doppler frequency shift observation would be different than Alice and Bob. 

However, in order to utilize Doppler frequency shift as a secrecy source, two important properties must be satisfied from the measurements:
\begin{enumerate}
\item The measurement of Eve should be as uncorrelated as possible  with the measurements at Alice and Bob. In other words, assuming Eve knows the secret key generation procedure, the generated key at her side must be different than Alice's and Bob's.

\item In order to obtain continuous secrecy, Alice and Bob would require update the generated key. Also, this generated key should have a  high amount of entropy to be utilized in symmetric key encryption. Therefore, Doppler frequency measurements should embody a random characteristics over a long time period.
\end{enumerate}
With the following proposition and the Observation 1, we guarantee the first given condition. The second condition is discussed in the following subsection.
\begin{proposition}
	For the nodes with very high mobility, Doppler frequency shift is a shared secret between reciprocal nodes, while any other eavesdropping node cannot retrieve the mobility information of the legitimate nodes. 
	\label{prop:not_eavesdrop} 
\end{proposition}
\begin{proof}
Let us denote the nominal Doppler frequency measurements at an arbitrary time instant at Alice and Bob respectively with $f_{ba}$ and $f_{ab}$. The measured Doppler frequency shifts at Eve would be denoted by $f_{be}$ and $f_{ae}$. Considering the definition of the velocity, 
\begin{equation}
\left((\mathbf{\hat{v}_{b}^T}-\mathbf{\hat{v}_{a}^T})+(\mathbf{\tilde{v}_{b}^T}-\mathbf{\tilde{v}_{a}^T})\right)\frac{(\mathbf{p_b}-\mathbf{p_{a}})}{||(\mathbf{p_b}-\mathbf{p_{a}})||}=f_{ba}.
\end{equation}
In order to obtain $f_{ba}$ or $f_{ab}$, Eve should obtain the vectors $\mathbf{\hat{v}_{a}}$, $\mathbf{\hat{v}_{b}}$, $\mathbf{p_{a}}$ and $\mathbf{p_{b}}$. The observations at Eve are consist of $f_{ae}$, $f_{be}$, $\mathbf{\hat{v}_{b}}$ and $\mathbf{\hat{v}_{a}}$. Let us rewrite the Doppler frequency equations at Eve in a set of linear equations form:
\begin{equation}
\begin{bmatrix}
\mathbf{p}^T_a-\mathbf{p}^T_e& -\mathbf{v}^T_e & 0 & 0\\ 
0& 0 & \mathbf{p}^T_b-\mathbf{p}^T_e & -\mathbf{v}^T_e
\end{bmatrix} \begin{bmatrix}
\mathbf{v}_a\\ 
\mathbf{p}_a\\ 
\mathbf{v}_b\\ 
\mathbf{p}_b
\end{bmatrix}= \begin{bmatrix}
f_{ae}-\mathbf{v}^T_e\mathbf{p}_e\\ 
f_{be}-\mathbf{v}^T_e\mathbf{p}_e
\end{bmatrix}.
\end{equation}
In order to capture the variable vector  $[\mathbf{v}_a, \mathbf{p}_a, \mathbf{v}_b , \mathbf{p}_b]^T$, Eve needs to find the inverse of the \textcolor{black}{ left part of the matrix}. In order to solve the equation Eve needs 6 more equations to find the inverse of the matrix given in the above. Therefore, Eve cannot retrieve the mobility information of Alice and Bob and cannot obtain their Doppler frequency measurements. 
\end{proof}

Note that, we do not consider the side channel attacks for the proposed scenario. Since the mobility of the nodes may obtain in different ways, Eve may try to utilize side channel attacks. For example, by utilizing multiple antenna receiver structure, Eve may estimate the position and the  velocity of the spacecrafts.  Moreover, by utilizing a radar system, she may track the legitimate nodes positions and velocities.  However, a radar equipped Eve would be limited by the range of the radar, the number of existing mobile nodes in the range, and the stealth technology capabilities of Alice and Bob. Based on the relative positions and velocities of Alice and Bob, they determine the unambiguous radar range for the eavesdropper as outlined in \cite{FOURIKIS20001}. If Alice and Bob are inside the radar range of Eve, than other existing mobile nodes would force Eve to try $\binom{n_s}{2}$ combinations for a single Doppler measurement, where $n_s$ denotes the number of mobile nodes in the radar range of Eve. By selecting the number of quantization intervals lower than $\binom{n_s}{2}$, brute force attack becomes easier than the radar attack. For example if there is at least one mobile node existing, they may continue their secret key generation with a two level quantization process. If there are no nodes other than Alice and Bob in the radar range, than they may utilize smart or cognitive jamming \cite{topal2020identification} or artificial noise methodologies to confuse the eavesdropper without effecting their own radar measurements \cite{magazin}. In the following, we discuss how the proposed mechanism satisfies the second enumerated condition for the secret key generation. 

\subsection{Maximum Achievable Secret Key Rate from Doppler Frequency}
The superposition of the velocities can also be observed in the Doppler frequency shifts as $\omega_{k\textcolor{black}{m}}(t)=\hat{\omega}_{k\textcolor{black}{m}}(t)+\tilde{\omega}_{k\textcolor{black}{m}}(t)$, where $\hat{\omega}_{k\textcolor{black}{m}}(t)=\frac{f_c}{c}\frac{(\mathbf{\hat{v}_{\textcolor{black}{m}}}(t)-\mathbf{\hat{v}_{k}}(t))^T(\mathbf{p_{\textcolor{black}{m}}}(t)-\mathbf{p_{k}}(t))}{||\mathbf{p_{\textcolor{black}{m}}}(t)-\mathbf{p_k}(t)||}$, and $\tilde{\omega}_{k\textcolor{black}{m}}(t)=\frac{f_c}{c}\frac{(\mathbf{\tilde{v}_{\textcolor{black}{m}}}(t)-\mathbf{\tilde{v}_{k}}(t))^T(\mathbf{p_{\textcolor{black}{m}}}(t)-\mathbf{p_{k}}(t))}{||\mathbf{p_{\textcolor{black}{m}}}(t)-\mathbf{p_k}(t)||}$. Note that, we model $\hat{\omega}_{k\textcolor{black}{m}}(t)$ as a deterministic constant, and $\tilde{\omega}_{k{\textcolor{black}{m}}}(t)$ as a random i.i.d. process. In the following, we show that, the deterministic part of the Doppler frequency shift actually do not contribute to the secret key randomness, whereas the random part determines the maximum achievable secret key rate. 
\begin{proposition}
The achievable secret key rate when the Doppler frequency shift is utilized as a secrecy source is equal to $R_k=\log({f_c}/{c})+h(\tilde{f}_{ab})$. 
\label{prop:rate}
\end{proposition}
\begin{proof}
The achievable secret key rate is limited by the entropy of the randomness source as given in \cite{elements}. In our case the achievable secret key rate is given by
\begin{equation}
R_k={h}(\omega_{ab})= {h}(\omega_{ba}).
\end{equation}
Let us rewrite the entropy of the secrecy source as 
\begin{equation}
{h}(\omega_{ab})= {h}(\hat{\omega}_{ab}+\tilde{\omega}_{ab}).
\end{equation}
As given in \cite{elements}, since the differential entropy is translation invariant, the constant term does not introduce additional entropy. Hence ${h}(\omega_{ab})= h(\tilde{\omega}_{ab})$. Also considering $\tilde{\omega}_{ab}=\frac{f_c \tilde{f}_{ab}}{c}$, the maximum achievable key rate is given by
\begin{equation}
R_k=\log({f_c}/{c})+h(\tilde{f}_{ab}).
\end{equation}

\end{proof}

\begin{proposition}
The achievable secret key rate is maximized when the mobility of the legitimate nodes follow the Brownian motion as  $\{(p_{a}(t+1)-p_a(t))\sim \mathcal{N}(0,\sigma_{v_a}^2T)$ and $(p_{b}(t+1)-p_b(t))\}\sim \mathcal{N}(0,\sigma_{v_b}^2T)$.
\label{prop:Brownian}
\end{proposition} 
\begin{proof}
Let us rephrase $R_k=\log({f_c}/{c})+h(\tilde{f}_{ab})$. 
Since the carrier frequency, speed of light are constants, the achievable secret key rate is maximized in line with the distribution of the nominal Doppler frequency. Since a Gaussian random variable has the largest entropy amongst all random variables of equal variance \cite{elements}, the maximum achievable secret key rate is obtained when $\tilde{f}_{ab}\sim \mathcal{N}(0,\sigma^2)$. $\tilde{f}_{ab}$ can be rephrased by 
\begin{equation}
\tilde{f}_{ab}= \rho_x(\tilde{v}^x_b-\tilde{v}^x_a)+\rho_y(\tilde{v}^y_b-\tilde{v}^y_a)+\rho_z(\tilde{v}^z_b-\tilde{v}^z_a),
\end{equation} 
where $\rho_x= \frac{p^x_b-p^x_a}{||p_b-p_a||}$, $\rho_y= \frac{p^y_b-p^y_a}{||p_b-p_a||}$, and $\rho_z= \frac{p^z_b-p^z_a}{||p_b-p_a||}$. By assuming each velocity component is Gaussian distributed, we can get the nominal Doppler frequency also Gaussian distributed, where $v^x_a, v^y_a, v^z_a\sim \mathcal{N}(0, \sigma_{v_a}^2)$, and $v^x_b, v^y_b, v^z_b \sim \mathcal{N}(0, \sigma_{v_b}^2)$.  The variance of the nominal Doppler frequency then becomes 
\begin{equation}
\sigma^2_d= (\rho^2_x+\rho^2_y+\rho^2_z)(\sigma_{v_a}^2+\sigma_{v_b}^2)= \sigma_{v_a}^2+\sigma_{v_b}^2.
\end{equation}
Since the velocity is constant in the communication period, the displacement vector of Alice and Bob would be modeled as $\{(p_{a}(t+1)-p_a(t))\sim \mathcal{N}(0,\sigma_{v_a}^2T)$ and $(p_{b}(t+1)-p_b(t))\}\sim \mathcal{N}(0,\sigma_{v_b}^2T)$, which is the definition of the Brownian motion \cite{elements}. 
\end{proof}

\begin{corollary}
The maximum achievable secret key rate can be given by
\begin{equation}
R_k^*=\frac{1}{2}\log(2\pi e \sigma^2_d)+\log(f_c/c). 
\end{equation}
\end{corollary}
\begin{proof}
This is a direct result of Proposition 2 and Proposition 3, where $$R^*_k=\log({f_c}/{c})+\max\{h(\tilde{f}_{ab})\},$$ and the maximization is obtained for $\tilde{f}_{ab}\sim \mathcal{N}(0,\sigma^2)$. Since $h(\tilde{f}_{ab})=\frac{1}{2}\log(2\pi e \sigma^2_d)$, the maximum achievable secret key rate becomes $$R_k^*=\frac{1}{2}\log(2\pi e \sigma^2_d)+\log(f_c/c).$$
\end{proof}

In the following, we present the proposed secret key generation procedure to extract secret keys from the mobility of the spacecrafts. 
\section{Secret Key Generation Procedure}
\begin{figure}[t]
	\centering
	\includegraphics[width=\linewidth]{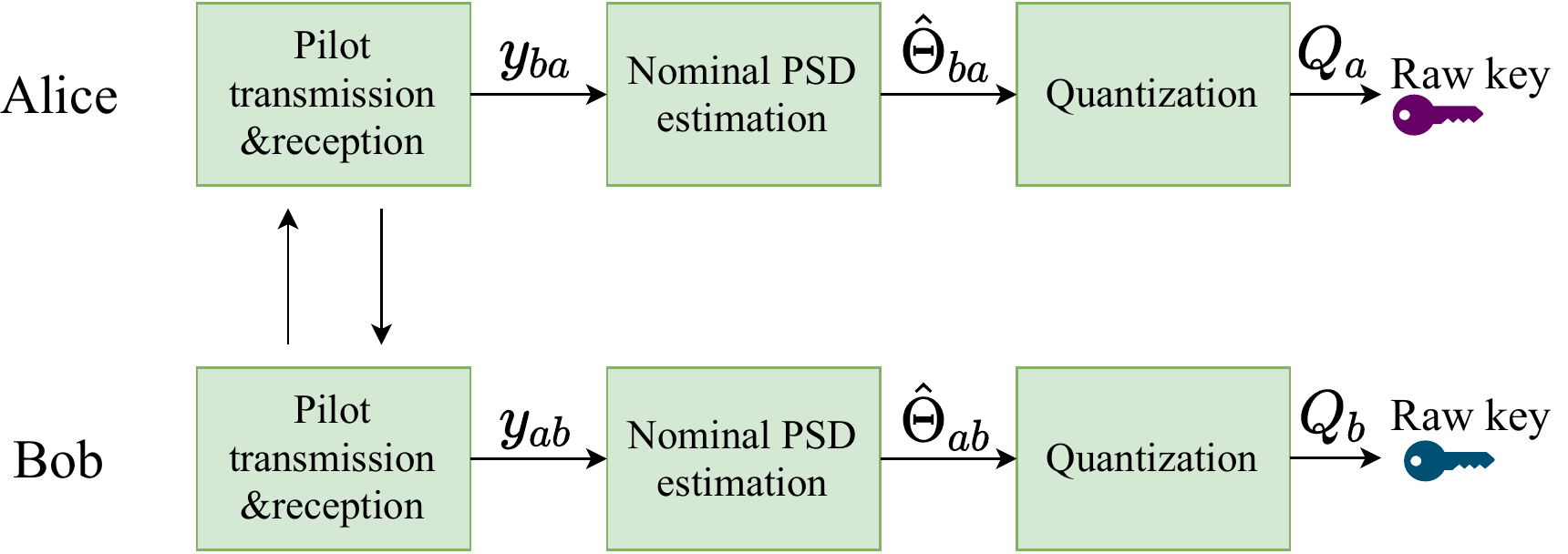}
	\caption{An illustration of secret key generation procedure.}
	\label{fig:system_model}
\end{figure}
\subsection{Pilot Transmission \& Reception}
The secret key generation procedure is described in the Figure \ref{fig:system_model}. As a first step, Alice and Bob respectively transmit $N$ pilot symbols $\mathbf{x}= \left[x(1), x(2), \ldots, x(N) \right]$ in a time division duplexing (TDD) fashion, {\color{black} where $x(i)=a_ie^{j\phi_i}$. In the following of this paper, we consider BPSK modulation for simulations and in Proposition 5. In this case, $a_i=1$, $\phi_i\in\{-\pi,\pi\}$.} The considered {\color{black}spacecraft}s Alice, Bob and Eve are respectively be denoted by $\{a,b,e\}$.   Considering the ISLs, the baseband representation of the received symbols at the {\color{black}spacecraft} $k$ from the {\color{black}spacecraft} $\textcolor{black}{m}$, $\mathbf{y_{\textcolor{black}{m}k}}=\left[y_{\textcolor{black}{m}k}(1), y_{\textcolor{black}{m}k}(2), \ldots, y_{\textcolor{black}{m}k}(N) \right]$ can be modeled as \
\begin{equation}
 y_{\textcolor{black}{m}k}(i)=x_{\textcolor{black}{m}k}(i)+\textcolor{black}{\epsilon}_{\textcolor{black}{m}k}(i),
\end{equation}
where $i\in\{1,\ldots,N\}$, $x_{\textcolor{black}{m}k}(i)=\zeta x(i)e^{j\omega_{\textcolor{black}{m}k}iT}$, $\zeta=\frac{1}{d^{PL}_{\textcolor{black}{m}k}}$ and $\textcolor{black}{m}\in\{a,b\}$, $k\in\{a,b,e\}$ and $\textcolor{black}{m}\neq k$ \cite{LEOsat}. $\zeta$ denotes the path loss attenuation, where $d_{\textcolor{black}{m}k}$ is the distance between spacecraft $\textcolor{black}{m}$ and $k$, and $PL$ is the path loss exponent. $\omega_{jk}$ denotes the Doppler frequency shift for the link $j-k$, and $T$ is the symbol period.  $\textcolor{black}{\epsilon}_{\textcolor{black}{m}k}(i)\sim \mathcal{CN}(0,\sigma_{k}^2)$ denotes the additive white Gaussian noise at the receiver $k$, where $\mathcal{CN}(0,\sigma^2)$ denotes the i.i.d. complex normal distribution with zero mean and $\sigma^2$ variance. Considering the similar quality of the radio equipments at each {\color{black}spacecraft}, we assume that the noise variance is equal for each receiver as $\sigma_{k}^2= \sigma^2$. { \color{black} Note that, we consider that the carrier frequency information is available at the receiver, where the received signal can be down-converted to the baseband equivalent signal. The effect of oscillator-based frequency shifts would be very small in comparison with the Doppler frequency shift from high mobility. We also assume that the channel based phase shift can be recovered at receiver by utilizing a phase locked loop as described in \cite{NASA_comm}.} {\color{black} The spacecrafts are assumed to have same velocities during the pilot transmission reception scheme. We detail the timing considerations in the Appendix. }

The information contained in $\mathbf{y_{\textcolor{black}{m}k}}$ is fully present in its discrete Fourier transform: 
\begin{equation}
Y_{\textcolor{black}{m}k}(i)=X_{\textcolor{black}{m}k}(i)+K_{\textcolor{black}{m}k}(i),
\end{equation}
where $\mathbf{Y_{\textcolor{black}{m}k}}=\mathcal{F}\{y_{\textcolor{black}{m}k}\}=\left[Y_{\textcolor{black}{m}k}(1), Y_{\textcolor{black}{m}k}(2), \ldots, Y_{\textcolor{black}{m}k}(N) \right]$, $\mathbf{X_{\textcolor{black}{m}k}}=\mathcal{F}\{x_{jk}\}=\left[X_{jk}(1), X_{jk}(2), \ldots, X_{jk}(N) \right]$ and $\mathbf{\textcolor{black}{\varepsilon}_{\textcolor{black}{m}k}}=\mathcal{F}\{\textcolor{black}{\epsilon}_{\textcolor{black}{m}k}\}=\left[\textcolor{black}{\varepsilon}_{\textcolor{black}{m}k}(1), \textcolor{black}{\varepsilon}_{\textcolor{black}{m}k}(2), \ldots, \textcolor{black}{\varepsilon}_{\textcolor{black}{m}k}(N) \right]$. Since Gaussian processes are invariant against Fourier transform, the signal spectrum $\mathbf{X_{\textcolor{black}{m}k}}$, and the noise spectrum $\mathbf{\textcolor{black}{\varepsilon}_{\textcolor{black}{m}k}}$, are also complex Gaussian, zero-mean, and orthogonal processes. The spectral samples ${Y_{\textcolor{black}{m}k}(i)}$ are mutually uncorrelated because of the assumed stationarity of $X_{\textcolor{black}{m}k}(i)$.

Note that, as stated in \cite{crbsystem}, the phase of $\mathbf{Y_{\textcolor{black}{m}k}}$ carries no information about the Doppler frequency, since $\mathbf{Y_{\textcolor{black}{m}k}}$ has been modeled as a stochastic process with the aforementioned properties. Hence, it is sufficient to consider the power spectrum of the received data as 

\begin{equation}
\begin{aligned}
\mathbf{S_{\textcolor{black}{m}k}} &= \left[S_{\textcolor{black}{m}k}(1), S_{\textcolor{black}{m}k}(2), \ldots, S_{\textcolor{black}{m}k}(N)\right] \\ 
&= \left[|Y_{\textcolor{black}{m}k}(1)|^2, |Y_{\textcolor{black}{m}k}(2)|^2, \ldots, |Y_{\textcolor{black}{m}k}(N)|^2\right].
\end{aligned}
\end{equation}

Since $Y_{\textcolor{black}{m}k}(i)$ is a complex Gaussian process, the probability density function of each sample ${S_{\textcolor{black}{m}k}}(i)$ under the condition of a particular Doppler frequency $\omega_{\textcolor{black}{m}k}$ is given by the exponential distribution :

\begin{equation}
\rho(S_{\textcolor{black}{m}k}(i);\omega_{\textcolor{black}{m}k})= \frac{1}{\Theta_{\textcolor{black}{m}k}(i)}\text{exp}\left(-\frac{S_{\textcolor{black}{m}k}(i)}{\Theta_{\textcolor{black}{m}k}(i)}\right), 
\label{eq:pdf}
\end{equation}
where $\Theta_{\textcolor{black}{m}k}(i)$ denotes the NPSDS, and can be obtained by \cite{crbsystem} 
\begin{equation}
\begin{aligned}
\Theta_{\textcolor{black}{m}k}(i)= \mathbb{E}\{{S_{\textcolor{black}{m}k}}(i)\} &= \mathbb{E}\{|X_{\textcolor{black}{m}k}(i)+\textcolor{black}{\varepsilon}_{\textcolor{black}{m}k}(i)|^2\} \\
&=\mathbb{E}\{|X_{\textcolor{black}{m}k}(i)|^2\}+\mathbb{E}\{|\textcolor{black}{\varepsilon}_{\textcolor{black}{m}k}(i)|^2\}
\end{aligned},
\end{equation}
where $A^x_{\textcolor{black}{m}k}(i\Delta f-\omega_{\textcolor{black}{m}k})=\mathbb{E}\{|X_{\textcolor{black}{m}k}(i)|^2\}$, $A^n_{\textcolor{black}{m}k}=\mathbb{E}\{|K_{\textcolor{black}{m}k}(i)|^2\}$. Note that $A^x_{\textcolor{black}{m}k}(f)$ is the a priori known nominal power spectral density of the signal; $\Delta f$ is the frequency sampling interval; and $\omega_{\textcolor{black}{m}k}$ is the Doppler frequency shift. Considering $A^x_{\textcolor{black}{m}k}(f)$ is periodic with period $\Delta f$, and $A^n_{\textcolor{black}{m}k}$ is a constant, we can deduce that $\Theta_{\textcolor{black}{m}k}(i)$ is also periodic with $\Delta f$, and consequently we can drop $i$ and denote the NPSDS as $$\Theta_{\textcolor{black}{m}k}=\Theta_{\textcolor{black}{m}k}(i), \forall i .$$

%\textbf{Observation 2:} \textit{The nominal power spectral density of a  }

\begin{proposition}
 The NPSDSs for reciprocal links will be equal to \begin{equation}
\Theta_{\textcolor{black}{m}k}=\Theta_{k\textcolor{black}{m}},
\end{equation}
while any other {\color{black}spacecraft} than $\textcolor{black}{m}$ and $k$ would observe different values.
\label{prop:NPSDS}
\end{proposition}

\begin{proof}
We can proof the Proposition 1 in two steps:
\begin{enumerate}
\item The NPSDS follows the symmetric relation as $A^x_{\textcolor{black}{m}k}(i\Delta f-\omega_{jk})= A^x_{\textcolor{black}{m}k}(i\Delta f+\omega_{\textcolor{black}{m}k})$ \cite{proakis}. As stated in Observation 1, $\omega_{\textcolor{black}{m}k}=-\omega_{k\textcolor{black}{m}}$. Considering the same pilot sequence is transmitted from both {\color{black}spacecrafts}, $A^x_{\textcolor{black}{m}k}(f)=A^x_{k\textcolor{black}{m}}(f)$.  Consequently, we can state that  $ A^x_{\textcolor{black}{m}k}(i\Delta f+\omega_{\textcolor{black}{m}k})= A^x_{k\textcolor{black}{m}}(i\Delta f+\omega_{\textcolor{black}{m}k})$.
\item As mentioned above, the variance of the thermal noise at each receiver is assumed to be equal. Therefore, $A^n_{\textcolor{black}{m}k}=A^n_{k\textcolor{black}{m}}$ and $$
\underset{\Theta_{\textcolor{black}{m}k}}{\underbrace{A^x_{\textcolor{black}{m}k}(i\Delta f+\omega_{\textcolor{black}{m}k})+A^n_{\textcolor{black}{m}k}}} =\underset{\Theta_{k\textcolor{black}{m}}}{\underbrace{A^x_{k\textcolor{black}{m}}(i\Delta f+\omega_{k\textcolor{black}{m}})+A^n_{k\textcolor{black}{m}}} } 
 $$
\end{enumerate}
\vspace{-2mm}
\end{proof}

As indicated in Proposition \ref{prop:NPSDS}, the ideal NPSDS at Alice and Bob would be equal for each communication block $\Theta_{ab}(t)=\Theta_{ba}(t)=\Theta_t$. However, considering the randomness in the mobility of the spacecrafts, the NPSDS are also random variables. Due to the dependency of the power spectral density to the modulation type, NPSDS distribution would also change with the modulation type. A general probability distribution model for the NPSDS become intractable. Therefore, we provide an approximation for  NPSDS modeling in the following proposition. 

\begin{proposition}
Considering BPSK modulation at the transmitter, the NPSDS at Alice and Bob can be modeled by $\frac{\Theta_t}{\sigma^2_{\Theta}T^2}\sim {\mathcal{\chi}'}^2_1(\frac{\hat{\omega}^2_{ab}}{\sigma^2_v})$, where $\sigma^2_{\Theta}=\frac{PT^3\sigma^2_v}{2}$. The probability  distribution function of ${\Theta_t}$ becomes
\begin{equation}
\rho(\Theta_t)=\frac{1}{2\sigma^2_{\Theta}}e^{-\frac{1}{2}\left(\frac{\Theta_t}{\sigma^2_{\Theta}}+\lambda\right)}\left({\frac{\Theta^2_t}{\sigma^2_{\Theta} \lambda}}\right)^{-\frac{1}{4}}\mathcal{I}_{-1/2}\left(\sqrt{\lambda\frac{\Theta^2_t}{\sigma^2_{\Theta}}}\right),
\end{equation}
where $\lambda=\frac{\hat{\omega}^2_{ab}}{\sigma^2_v}$.
\label{prop:psd}
\end{proposition}
\begin{proof}
As we modeled in Section II, the Doppler frequency shift can be written as the superposition of two components, $\omega_{ab}(t)=\hat{\omega}_{ab}(t)+\tilde{\omega}_{ab}(t)$, where $\tilde{\omega}_{ab}(t)\sim\mathcal{N}(0,\sigma^2_v)$. As given in \cite{proakis}, considering the BPSK modulation NPSDS becomes 
\begin{equation}
\Theta_t=\frac{PT}{2}\text{sinc}^2(1-\omega_{ab}T),
\end{equation}
where $P$ denotes the average transmission signal power, $T$ denotes the symbol period. 
By utilizing the Taylor series expansion of $\text{sinc}^2(x)=(x-1)^2+\mathcal{O}(x^4)$ as $x\rightarrow 1$, the NPSDS becomes 
\begin{equation}
\Theta_t= \frac{PT^3}{2}\left({\omega_{ab}}\right)^2. 
\end{equation}
We can model the Doppler frequency shift as $\omega_{ab}\sim\mathcal{N}(\hat{\omega}_{ab}, \sigma_v^2)$. Considering this, the random variable $\zeta= \frac{\Theta_t}{\sigma^2_v\frac{PT^3}{2}}$ becomes a non-central chi-squared distribution with 1 degrees of freedom and $\lambda=\frac{\hat{\omega}^2_{ab}}{\sigma^2_v}$ non-centrality parameter. Since , ${\Theta_t}=\zeta{\sigma^2_v\frac{PT^3}{2}}$, the probability distribution function of $\Theta_t$ can be obtained by the change of the variables as $$\rho(\Theta_t)=\frac{1}{2\sigma^2_{\Theta}}e^{-\frac{1}{2}\left(\frac{\Theta_t}{\sigma^2_{\Theta}}+\lambda\right)}\left({\frac{\Theta^2_t}{\sigma^2_{\Theta} \lambda}}\right)^{-\frac{1}{4}}\mathcal{I}_{-1/2}\left(\sqrt{\lambda\frac{\Theta^2_t}{\sigma^2_{\Theta}}}\right).$$ 
\end{proof}
Note that, the given distribution in Proposition \ref{prop:psd} applies for the BPSK modulation. With different modulation type, the shape of the power spectral density would also be different. In Section VI, we provide the tightness of the assumptions made in the numerical analyses on the Proposition \ref{prop:psd}. Without explicitly calculating the Doppler frequency, {\color{black}spacecrafts} may individually estimate the NPSDS from their observation sequences, and they can exploit this value as a shared secret for key generation mechanism. In the following, we describe the estimation of the NPSDS process at the {\color{black}spacecraft}s.
\subsection{NPSDS Estimation}
Due to the NPSDS change in each communication block, the NPSDS estimation would be repeated for each communication block separately.  Let us denote the estimated NPSDS at receiving node $k$ in an arbitrary communication block as ${\Theta}_{\textcolor{black}{m}k}$. Note that considering Alice and Bob, $\Theta_{ab}=\Theta_{ba}=\Theta_t$. The maximum likelihood (ML) estimation of the parameter ${\Theta}_{\textcolor{black}{m}k}$ can be given as:

\begin{equation}
\begin{aligned}
\hat{\Theta}_{\textcolor{black}{m}k} &= \max_{{\Theta}_{\textcolor{black}{m}k}}\left\{ \prod\limits_{i=1}^{N} \frac{1}{{\Theta}_{\textcolor{black}{m}k}}\text{exp}\left(-\frac{S_{\textcolor{black}{m}k}(i)}{\Theta_{\textcolor{black}{m}k}}\right) \right\} \\
&=\max_{{\Theta}_{\textcolor{black}{m}k}}\left\{\frac{1}{{\Theta}^N_{\textcolor{black}{m}k}}\text{exp}\left(-\frac{\sum\limits_{i=1}^{N}S_{\textcolor{black}{m}k}(i)}{\Theta_{\textcolor{black}{m}k}}\right) \right\}. 
\end{aligned}
\end{equation}

The log-likelihood function for the estimation problem becomes
\begin{equation}
\mathcal{L}(S_{\textcolor{black}{m}k}(i);{\Theta}_{\textcolor{black}{m}k})= -N\ln({\Theta}_{\textcolor{black}{m}k})-\left(\frac{\sum\limits_{i=1}^{N}S_{\textcolor{black}{m}k}(i)}{\Theta_{\textcolor{black}{m}k}}\right).
\end{equation}

\textcolor{black}{We obtain the value of the function when its first derivative is zero to find the maximum value of the log-likelihood function}

\begin{equation}
\frac{\partial \mathcal{L}(S_{\textcolor{black}{m}k}(i);{\Theta}_{\textcolor{black}{m}k})}{\partial \Theta_{\textcolor{black}{m}k}}= -\frac{N}{{\Theta}_{\textcolor{black}{m}k}}+\frac{\sum\limits_{i=1}^{N}S_{\textcolor{black}{m}k}(i)}{\Theta^2_{\textcolor{black}{m}k}}=0.
\end{equation}
The result of this equation provides us the ML estimator for the ${\Theta}_{\textcolor{black}{m}k}$ parameter as 
\begin{equation}
\hat{\Theta}_{\textcolor{black}{m}k}= M^{\textcolor{black}{m}k}_s= \frac{\sum\limits_{i=1}^{N}S_{\textcolor{black}{m}k}(i)}{N},
\label{eq:estimation}
\end{equation}
where $ M^{\textcolor{black}{m}k}_s$ is the sample mean of the observed power spectral density samples. The steps followed in NPSDS estimation block can be summarized as in Figure \ref{fig:psd}.
\begin{figure}[t]
	\centering
	\includegraphics[width=\linewidth]{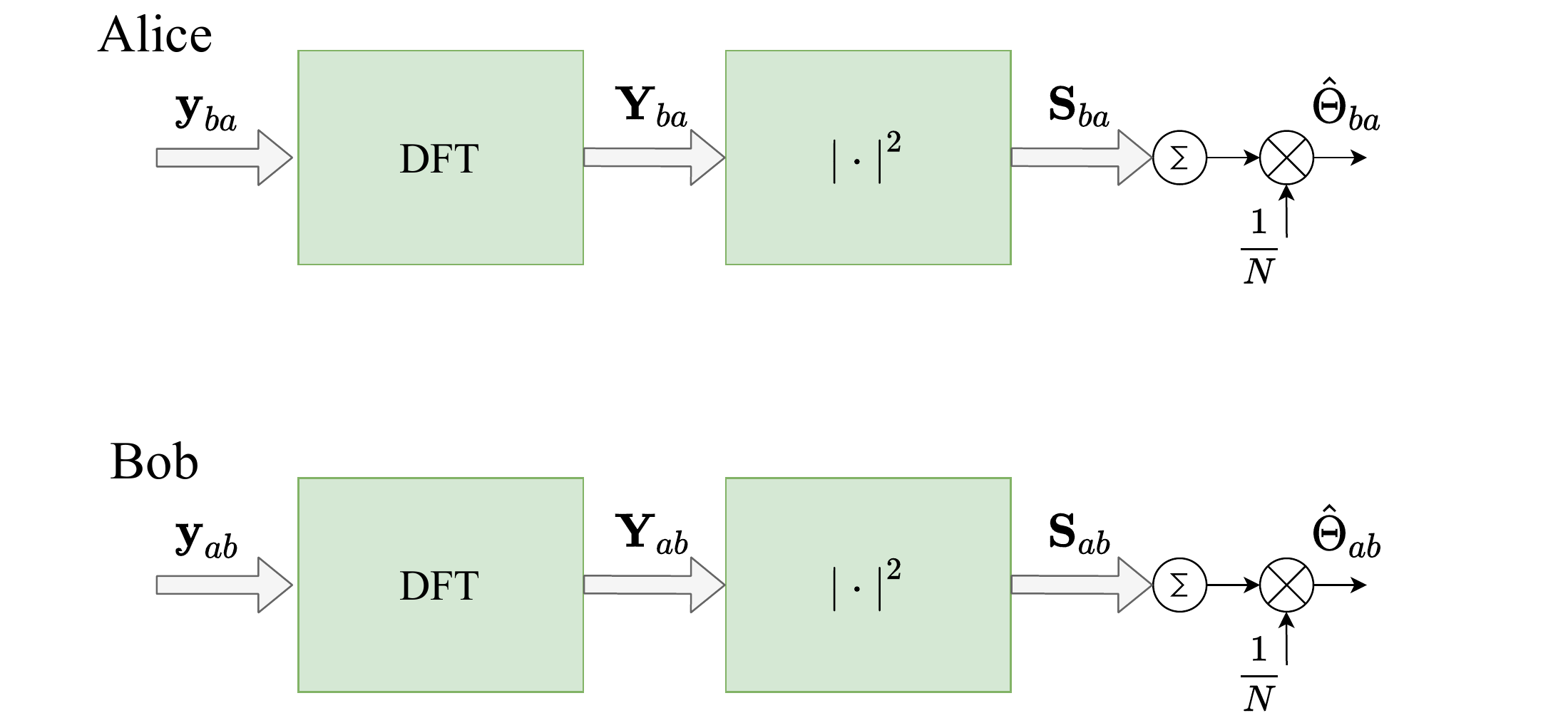}
	\caption{An illustration of NPSDS estimation at Alice and Bob.}
	\label{fig:psd}
\end{figure}
\subsection{Quantization}
After obtaining  $\hat{\Theta}_{\textcolor{black}{m}k}$ values, terminals quantize and encode this value in order to obtain a secret key sequence. In this paper, we assume uniform quantizer with a step size $\Delta$ as in \cite{Topal}. From $N$ observations, we get $Q$ number of quantized secret key bits as $q_{b}= \left \lfloor \frac{\hat{\Theta}_{ab}}{\Delta} \right \rfloor$  and $q_{a}= \left \lfloor \frac{\hat{\Theta}_{ba}}{\Delta} \right \rfloor$, where $ \left \lfloor \cdot \right \rfloor$ denotes the floor function.

Note that, the main focus of this work is producing the raw key from a novel secrecy source, mainly the Doppler frequency shift. The key reconciliation and hybrid key generation algorithms in  \cite{hybrid} can be applied to the quantized raw key bits in order to reduce the erroneous elements in the generated key bits.

\section{Key Disagreement Rate}
Key disagreement rate (KDR) is the ratio of mismatched bits in the generated keys. We adopt a similar approach to \cite{WiSee} in order to obtain theoretical expressions of the KDR considering the proposed key generation system model.  %To obtain the key disagreement rate, we normalize the estimated NPSDSs by multiplying with a normalization constant $ \eta=N/\Theta_{ab} $  as $\tilde{\Theta}_{ab}=\eta\hat{\Theta}_{ab}$ and $\tilde{\Theta}_{ba}=\eta\hat{\Theta}_{ba}$. 
Considering Proposition \ref{prop:NPSDS}, we can state that $\Theta_{ba}=\Theta_{ab}=\Theta_t$ for the $t^{\text{th}}$ communication block. Note that, one main difference with our previous work is that, we model $\Theta_t$ with a random variable as given in the Proposition \ref{prop:psd}. In the following, we approximate the estimated NPSDS values at Alice and Bob by using the central limit theorem.  
\begin{proposition}
For large $N$, NPSDS estimates follow the truncated normal distribution with support $[0,\infty]$ as $\hat{\Theta}_{ab}\sim \mathcal{N}(\Theta_t, \frac{\Theta^2_t}{N})$, and $\hat{\Theta}_{ba}\sim \mathcal{N}(\Theta_t, \frac{\Theta^2_t}{N})$.
\label{prop:central_limit}
\end{proposition}
\begin{proof}
Let us rewrite $\hat{\Theta}_{\textcolor{black}{m}k}= M^{\textcolor{black}{m}k}_s= \frac{\sum\limits_{i=1}^{N}S_{\textcolor{black}{m}k}(i)}{N}$. Here, $S_{\textcolor{black}{m}k}(i)$ samples are i.i.d, and each one of them follows the probability distribution function given in \textcolor{black}{Eq. \ref{eq:pdf}}. Considering this, for large $N$, the central limit theorem states that $\sqrt{N}\left(\hat{\Theta}_{ab}- \Theta_t\right)\sim \mathcal{N}(0,\Theta^2_t)$ \cite{proakis}. Consequently $\hat{\Theta}_{ab}\sim \mathcal{N}(\Theta_t, \frac{\Theta^2_t}{N})$. Identical operations can be applied on $\hat{\Theta}_{ba}$, due to $\Theta_{ab}=\Theta_{ba}=\Theta_t$. Note that, the NPSDS should be larger than equal to zero from its definition, hence their distribution can be given as the truncated normal distribution with the support of $[0,\infty]$ \cite{kotz2004continuous}. 
\end{proof}

Assuming the uniform quantizer with a step size $\Delta$, an estimated $\hat{\Theta}_{ab}$ is mapped to the $l^{\text{th}}$ quantization level, where $l= \left \lfloor \frac{\hat{\Theta}_{ab}}{\Delta} \right \rfloor$ and the quantization interval is described by $I_l=[l\Delta, (l+1)\Delta]$. The probability that $\hat{\Theta}_{ba}$ locates in  $I_l$ under given $\tilde{\Theta}_{ab}$ can be obtained by
\begin{equation}
P_l=\int\limits_{l\Delta}^{(l+1)\Delta}\rho(\hat{\Theta}_{ba}|\hat{\Theta}_{ab},\Theta_t)d\hat{\Theta}_{ba},
\end{equation}
where $\rho(\hat{\Theta}_{ba}|\hat{\Theta}_{ab},\Theta_t)$ denotes the probability density function of $\hat{\Theta}_{ba}$ given $\hat{\Theta}_{ab}$ and $\Theta_t$. In order to obtain $P_l$, let us first obtain $\rho(\hat{\Theta}_{ba}|\hat{\Theta}_{ab},\Theta_t)$.
For a given $\hat{\Theta}_{ab}$, the observation signal at Alice can be modeled as $\hat{\Theta}_{ba}=\hat{\Theta}_{ab}-{K}_{ab}+{K}_{ba}$, where $\{{K}_{ab},{K}_{ba}\}\sim \mathcal{N}(0,\frac{\Theta^2_t}{N})$.  In this case, we can state that $\hat{\Theta}_{ba} \sim \mathcal{N}(\hat{\Theta}_{ab},\frac{2\Theta^2_t}{N})$. Therefore $P_l$ becomes
\begin{equation}
P_l=\!\!\!\!\!\!\int\limits_{l\Delta}^{(l+1)\Delta} \frac{\sqrt{N}}{\sqrt{8\pi }\Theta_t} e^{-\frac{N(\hat{\Theta}_{ba}-\hat{\Theta}_{ab})^2}{8\Theta^2_t}} d\hat{\Theta}_{ba},
\end{equation}
The closed form expression for the $P_l$ can be expressed by 
\begin{equation}
P_l= \text{erf}\left[\frac{(l+1)\Delta-\hat{\Theta}_{ab}}{\sqrt{\frac{8\Theta^2_t}{N}}}\right]- \text{erf}\left[\frac{(l)\Delta-\hat{\Theta}_{ab}}{\sqrt{\frac{8\Theta^2_t}{N}}}\right]. 
\end{equation}
Considering that $\hat{\Theta}_{ab}$ is also a random variable, the key matching probability can be obtained by
\begin{equation}
P_c=\int_{0}^{\infty}\left(\int_{0}^{\infty}P_l \rho(\tilde{\Theta}_{ab}|\Theta_t) d\tilde{\Theta}_{ab}\right)\rho(\Theta_t)d\Theta_t.
\end{equation} 

To the best of authors knowledge, a closed form expression for the key matching probability is not available. In the most open form, the key matching probability becomes as in (\ref{eq:final}). In the following, we provide the numerical analysis of the proposed key generation mechanism. 
\begin{figure*}[t]
	\footnotesize
\begin{equation}
P_c=\int_{0}^{\infty}\left(\int_{0}^{\infty}\text{erf}\left[\frac{(l+1)\Delta-\hat{\Theta}_{ab}}{\sqrt{\frac{8\Theta^2_t}{N}}}\right]- \text{erf}\left[\frac{(l)\Delta-\hat{\Theta}_{ab}}{\sqrt{\frac{8\Theta^2_t}{N}}}\right]\frac{\sqrt{N}}{\sqrt{2\pi }\Theta_t}e^{-\frac{N(\hat{\Theta}_{ab}-{\Theta}_{t})^2}{2\Theta^2_t}}d\hat{\Theta}_{ab}\right)\frac{1}{2\sigma^2_{\Theta}}e^{-\frac{1}{2}\left(\frac{\Theta_t}{\sigma^2_{\Theta}}+\lambda\right)}\left({\frac{\Theta^2_t}{\sigma^2_{\Theta} \lambda}}\right)^{-\frac{1}{4}}\mathcal{I}_{-1/2}\left(\sqrt{\lambda\frac{\Theta^2_t}{\sigma^2_{\Theta}}}\right)d\Theta_t.
\label{eq:final}
\end{equation}
\end{figure*}

\section{Numerical Analysis}
In this section, we first analyze the secret key rate that can be extracted from the Doppler frequency measurements in the considered scenario. In the following, we provide the numerical analysis on the proposed secret key generation mechanism based on the estimated NPSDSs. \textcolor{black}{We conduct the Monte Carlo simulations on MATLAB.}
\subsection{Maximum Achievable Key Rate from Doppler Frequency Shift Analysis}
In this part, we numerically analyze the results given in Proposition \ref{prop:Brownian} and in Proposition \ref{prop:rate}. The velocity model at Alice and Bob follows the model given in the Section II.  In order to illustrate the importance of the random velocity variable in key generation, we define mobility constant at Alice as $\textcolor{black}{\kappa}_a= \frac{\sigma^2_v}{||\hat{\mathbf{v}}_a||}$, and similarly at Bob as $\textcolor{black}{\kappa}_b= \frac{\sigma^2_v}{||\hat{\mathbf{v}}_b||}$. The parameters for the numerical analyses given in Table \ref{tab:simulation_parameters}. 

Figure \ref{fig:key_rate} illustrates the maximum achievable secret key rate values for different mobility constants at Alice and Bob. One observation is that the secret key rate gradually increases as any one of the mobility constants increases. Since the randomness of the secret key is directly related with the random movement of Alice and Bob, the  secret key rate increases as the contribution of the random movement increases in any one of the nodes. Furthermore, the secret key rate gets its maximum value as the random velocity component at Alice and Bob becomes more dominant to the constant velocity component. However, this kind of setting is not realizable considering the space missions requirements. Therefore, we assume that the lower $k$ region would illustrate more realistically the future space missions.  Note that in line with the designed secret key generation procedure, the secret key rate changes. The values given here is an upper bound for the maximum amount of secret key can be extracted from the Doppler frequency shift. 
\begin{figure}[tb]
\centering
\includegraphics[width=\linewidth]{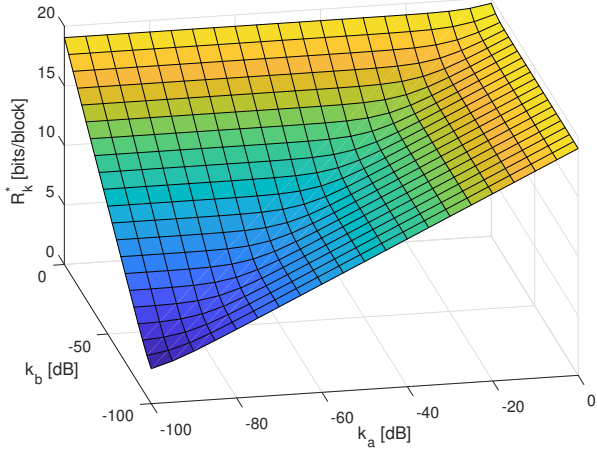}
\caption{Maximum achievable secret key rate for varying mobility constants of Alice and Bob.}
\label{fig:key_rate}
\end{figure}

\subsection{NPSDS Modeling Analysis}
In the following, we present the numerical analysis related with the proposed secret key generation methodology.  
%\begin{figure}[t]%!b
%	\centering	
%	\subfigure[]{
%		\label{fig:k1}
%		\includegraphics[width=0.46\linewidth]{k01.eps} }
%	\subfigure[]{
%		\label{fig:k2}
%		\includegraphics[width=0.46\linewidth]{k025.eps} }
%	\subfigure[]{
%		\label{fig:k3}
%		\includegraphics[width=0.46\linewidth]{k05.eps} }
%	\subfigure[]{
%		\label{fig:k4}
%		\includegraphics[width=0.46\linewidth]{k08.eps} }
%	\vspace{-0.4cm}
%	\caption{Empirical probability density functions for NPSDSs for three different $k$ values, (a) $k=0.1$, (b) $k=0.25$ and (c) $k=0.5$, (d) $k=0.8$.}
%	\label{fig:pdfs}
%\end{figure}
\begin{figure}[t]
	\centering
	\includegraphics[width=\linewidth]{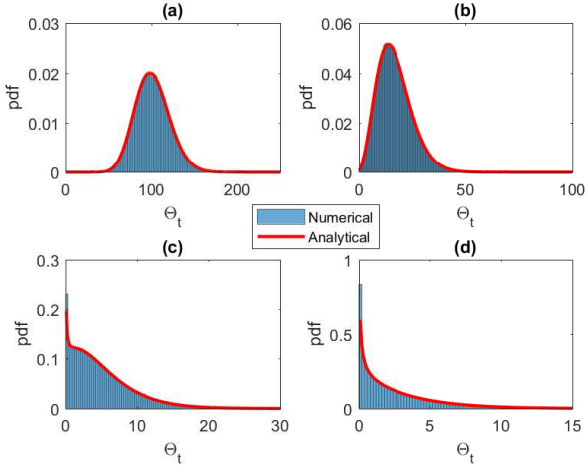}
	\caption{Empirical probability density functions for NPSDSs for three different $k$ values, (a) $k=0.1$, (b) $k=0.25$ and (c) $k=0.5$, (d) $k=0.8$.}
	\label{fig:pdfs}
\end{figure}

 Figure \ref{fig:pdfs} provides a comparison on the probability density functions of NPSDSs when the value is calculated numerically, and when the value is calculated by the approximation given in (19).The simulation parameters are given in Table \ref{tab:simulation_parameters}. Note that in this simulation we assume $\textcolor{black}{\kappa}_a=\textcolor{black}{\kappa}_b=\textcolor{black}{\kappa}$, where $0\leq \textcolor{black}{\kappa} <1$. One important observation is that, the center point of the distribution shifts near to zero as the random competent becomes dominant. One reason for this trend is that as the $\textcolor{black}{\kappa}$ increases $\sigma^2_\Theta$ also increases. Because of this, the non-centrality  parameter of the $\Theta_t$ decreases and the center of the distribution gets closer to zero. From another perspective, as the $\textcolor{black}{\kappa}$ increases total amount of movement and Doppler frequency increases. By increasing these, the samples from the nominal power spectral densities would be taken from further than their center and maximum values. Considering the shape of the power spectral density of the modulated waveforms, as the sample points get farther than the center, the amplitude values decreases. Therefore, the center of the distribution tends to shift near zero.

\subsection{NPSDS Estimation Analysis}
\begin{figure*}[t]%!b
	\begin{center}	
		\begin{adjustbox}{max width=\textwidth}	
			\subfigure[]{
				\label{fig:pdfN10}
				\includegraphics[width=\textwidth]{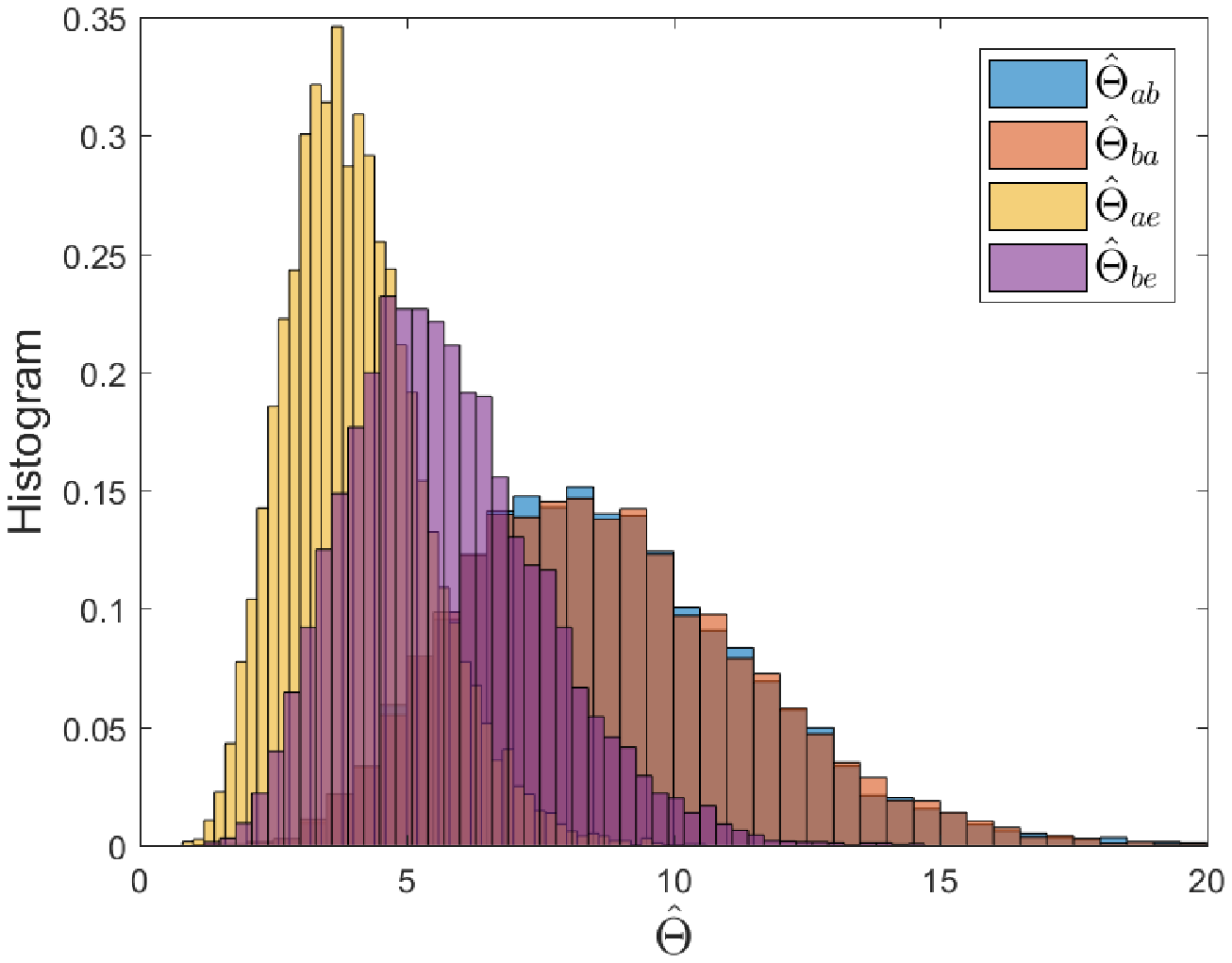} }
			\hfill
			
			\subfigure[]{
				\label{fig:pdfN20}
				\includegraphics[width=\textwidth]{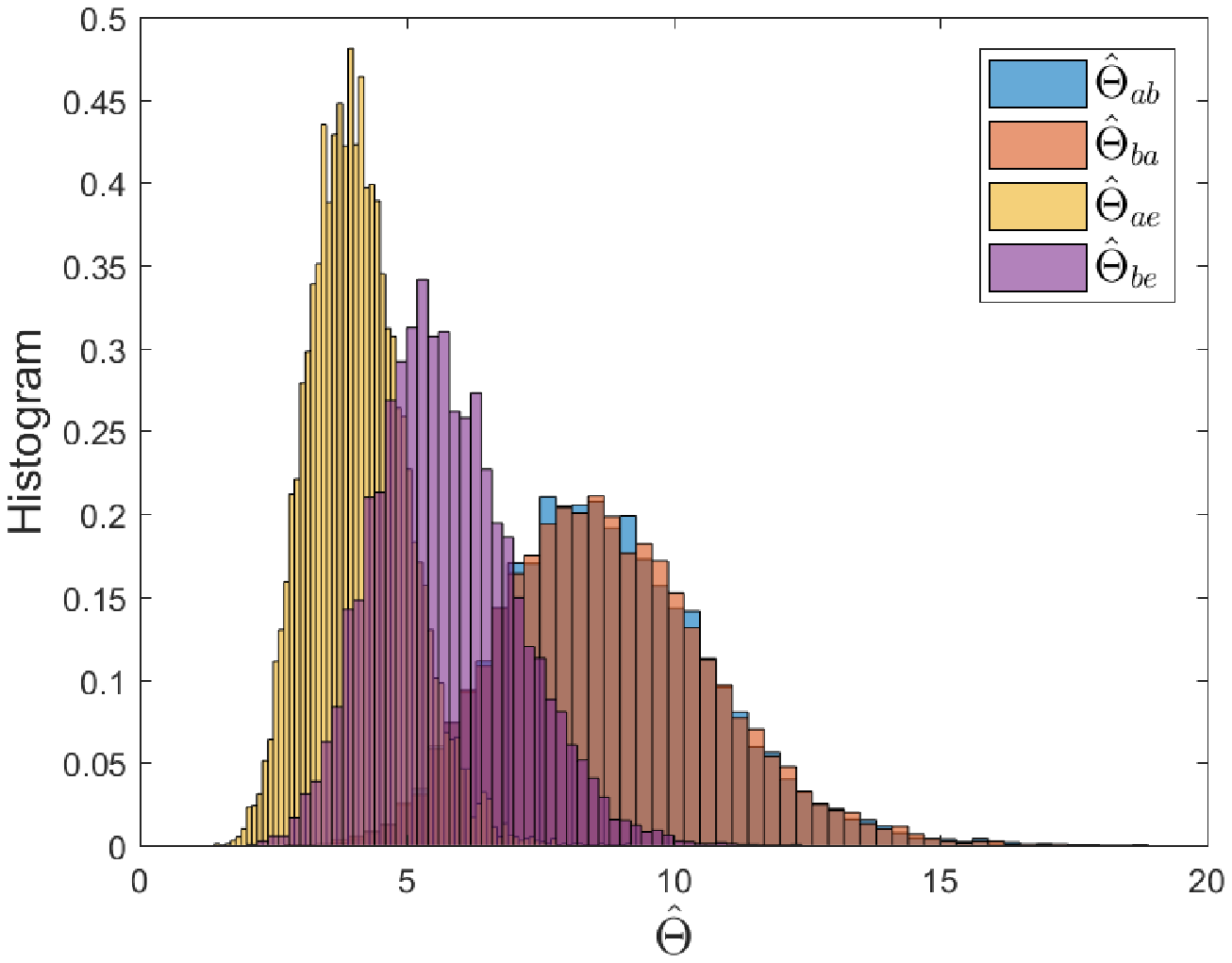} }
			\hfill
			\subfigure[]{
				\label{fig:pdfN50}
				\includegraphics[width=\textwidth]{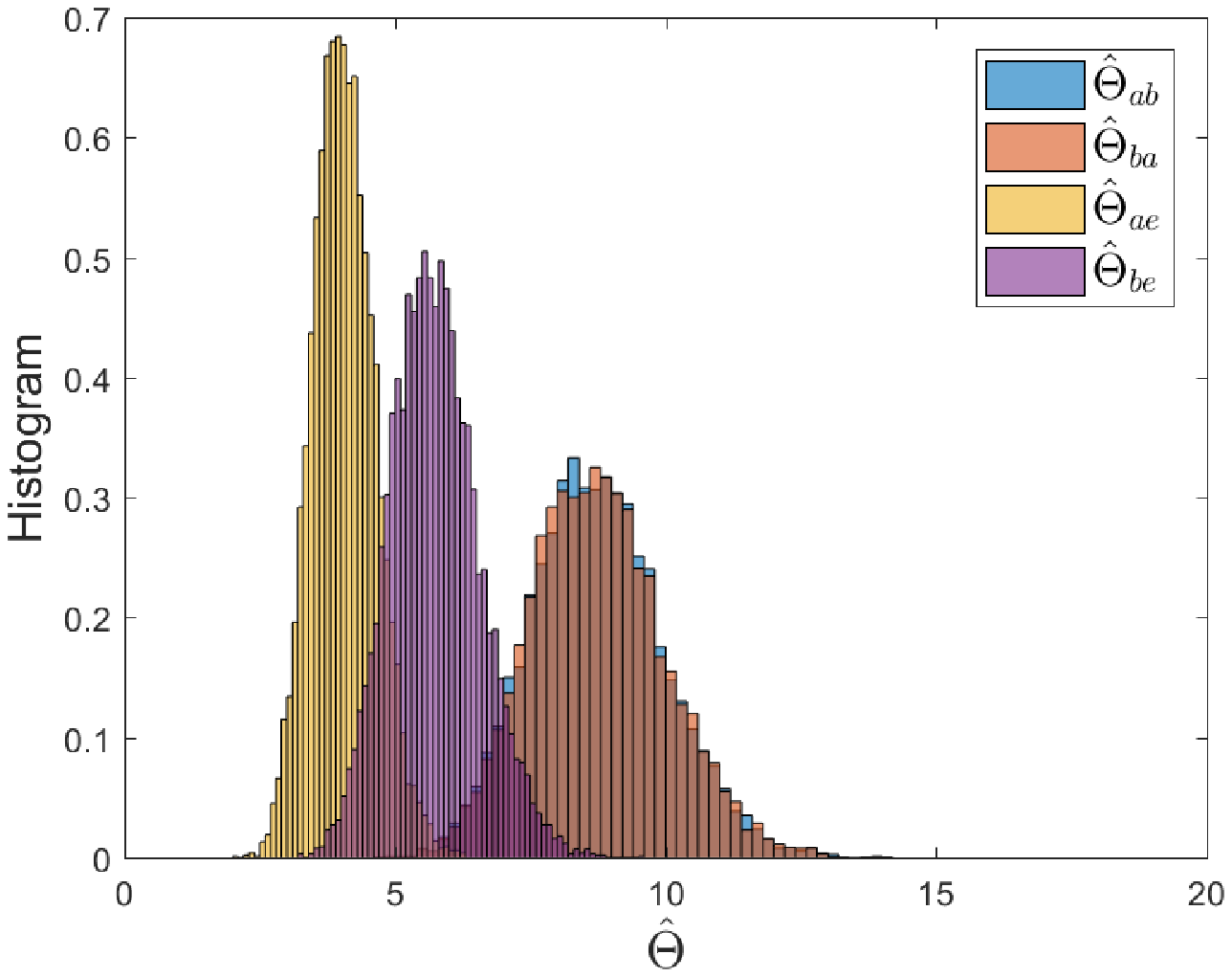} }
		\end{adjustbox}
	\end{center}
	\caption{Histograms for estimated NPSDSs for three different $N$ values, (a) $N=10$, (b) $N=20$ and (c) $N=50$.}
	\label{fig:MixHist-sc1}
\end{figure*}
Alice and Bob employ the proposed key generation scheme in Section II in order to obtain a secret key. In a single key duration, first Alice then Bob respectively transmits $N$ pilot symbols. Then, they estimate NPSDSs, and feed them into uniform quantizer. In the meantime, Eve observes the transmitted pilot symbols from both {\color{black}spacecraft}s, and obtain NPSDSs for each transmission. After each communication block, Alice and Bob update their mobility information, and consequently their Doppler frequency based observations are also updated. Each of the legitimate nodes generate single secret key, while Eve generates two different versions of the secret key. 
% The simulation parameters are given in Table \ref{tab:simulation_parameters}, where the values are obtained from \cite{numerical}. The NPSDS of BPSK simulation is utilized as in \cite{proakis}. {\color{black} Note that, as the mobility of two nodes continiues to change over time as in the {\color{black}spacecraft}s, the proposed key generation mechanism can also be applied at higher frequency bandwidths.}

\begin{table}[t]
	\caption{Simulation parameters}
	\centering
	\begin{tabular}{|l|l|}
		\hline
		Carrier frequency  & 1 GHz   \\ \hline
		$||\hat{\mathbf{v}}_{a}||$ , $||\hat{\mathbf{v}}_{b}||$      & 6700 m/sc \\ \hline
		${v}_{e}$      & 2000 m/sc \\ \hline
		Symbol energy      & 10 dB   \\ \hline
		Modulation type    & BPSK    \\ \hline
		Noise variance     & 1 dB    \\ \hline
	    Symbol period     & $10^{-6}$ sc.    \\ \hline
	\end{tabular}
	\label{tab:simulation_parameters}
\end{table}

\begin{figure}[h]
	\centering
	\includegraphics[width=\linewidth]{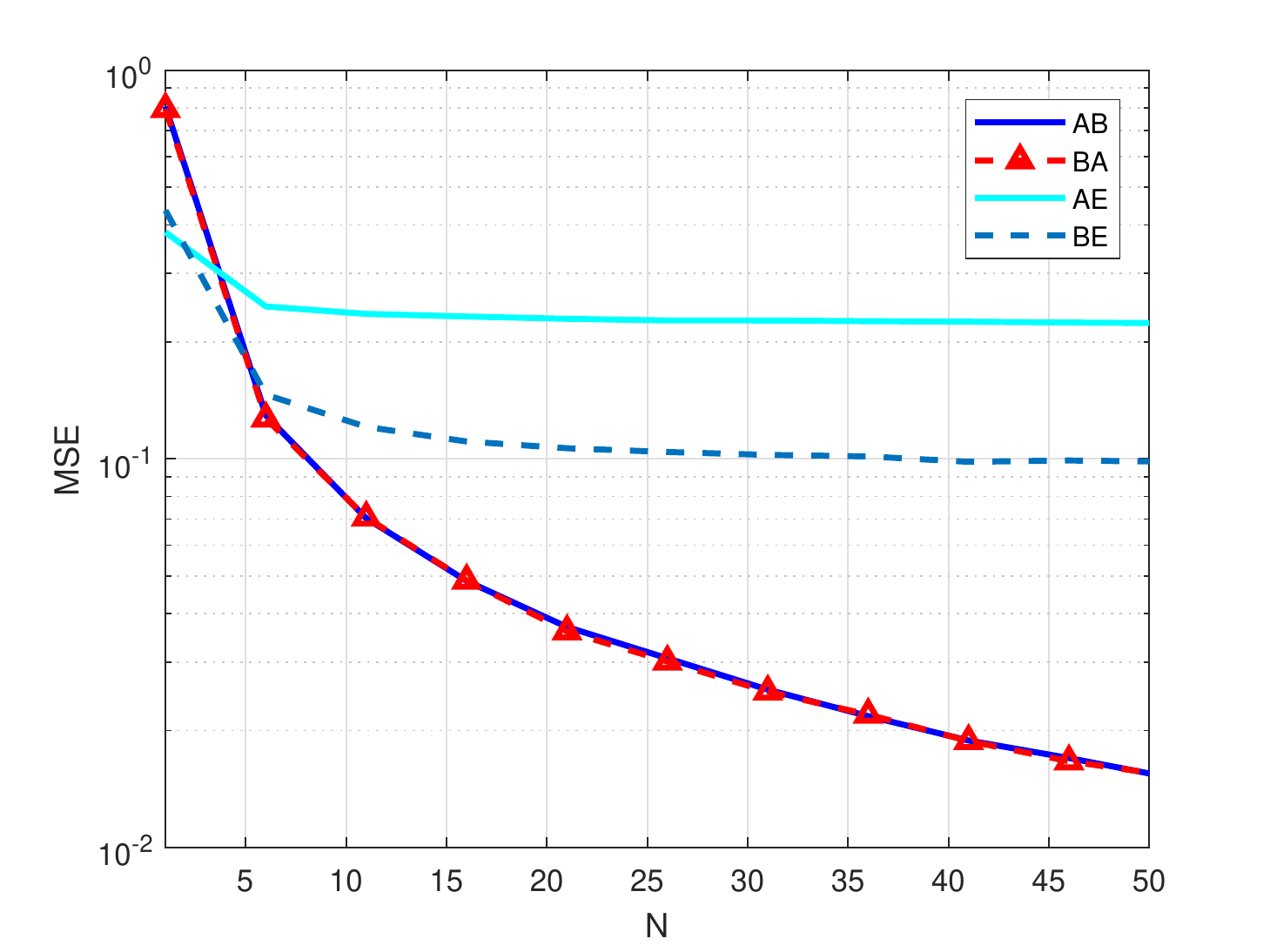}
	\caption{MSE vs. channel observations considering the different ISLs.}
	\label{fig:MSE}
\end{figure}

Figure \ref{fig:MixHist-sc1} shows the histograms of the estimated NPSDSs at Alice, Bob and Eve for a given $\Theta_t=5$. Since Eve captures both the transmission from Alice to Bob and from Bob to Alice, she obtains two different versions of the estimated NPSDSs. From left to right (Figure \ref{fig:pdfN10}-\ref{fig:pdfN20}-\ref{fig:pdfN50}), the number of observations increases. One remark is that as the number of observation increases, the variance of the pdfs become smaller. The distribution at Alice and Bob are almost identical as defined in (4) and (5). On the other hand, the pdfs at Eve diverges from Alice and Bob, since their relative velocities and Doppler frequencies are different. 

In order to focus on the relationship between the number of observations and the estimation disparities, we utilize mean squared error (MSE) as in Figure \ref{fig:MSE}. The MSE for the NPSDS  estimator for the $\textcolor{black}{m}-k$ link can be given as $$\text{MSE}_{\textcolor{black}{m}k}=\frac{1}{N} \sum_{i=0}^{N}|\hat{\Theta}_{\textcolor{black}{m}k}-\Theta_{t}|^2.$$ As indicated in the figure, the estimations at Alice and Bob converges as the number of observations increases. Even though estimation error at Eve also decreases, the error becomes stable after $N=5$. Since the relative velocity of Alice and Bob differs from the relative velocity of Alice and Eve, and the relative velocity of Bob and Eve, the Doppler frequency values observed at Eve is also different than Alice and Bob. Considering the Proposition \ref{prop:not_eavesdrop}, the only option for Eve is to utilize her estimation to obtain the secret key generated by Alice and Bob since she cannot retrieve the mobility information of Alice and Bob from her observations. The results in both Figure \ref{fig:MixHist-sc1} and Figure  \ref{fig:MSE} show that the NPSDSs or Doppler frequencies can be utilized as a secret source between two separately moving nodes.

\begin{figure}[tbh]
	\centering
	\includegraphics[width=\linewidth]{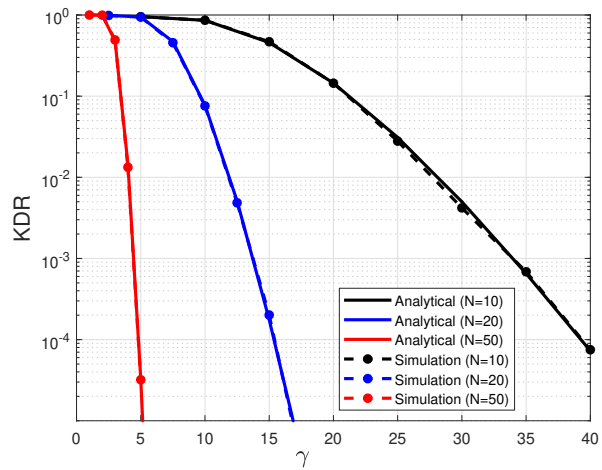}
	\caption{KDR vs. normalized quantization interval considering different number of channel taps.}
	\label{fig:KDR}
\end{figure}

\subsection{KDR Analysis}
KDR analysis illustrates the robustness of the secret key scheme. Since even one bit mismatch in the generated secret key results non-identical secret keys at Alice and Bob, the error rates should be very low to practically apply the secret key scheme. Numerically KDR in a single key duration can be described by
$$
KDR_{\textcolor{black}{m}k}= \begin{cases}
1, q_{\textcolor{black}{m}}\neq q_k \\ 0, q_{\textcolor{black}{m}}=q_k
\end{cases}.
$$
Considering $D$ number of key durations, KDR for {\color{black}spacecraft}s $k$ and $j$ becomes
$$
\text{KDR}_{\textcolor{black}{m}k}= \frac{\sum\limits_{t=1}^{D}KDR_{\textcolor{black}{m}k}(t)}{D}.
$$
Figure \ref{fig:KDR} provides KDR values of the generated keys at Alice and Bob for different normalized quantization intervals, where $\Gamma=\frac{\Delta}{\mathbb{E}\{\Theta_t\}}$. We consider a two-level uniform quantizer for the proposed secret key scheme. So the NPSDS estimations are mapped into two different symbols. Since as the number of observations at Alice and Bob increases, the error at their estimates decreases. Therefore, increment in the number of observation samples decreases KDR more dramatically. Note that, increment in the number of observation samples would also require transmitting more pilot symbols and reduces the efficiency of the protocol. Solid lines show the approximate KDR by (\ref{eq:final}) while the stars indicate the simulation results. Note that, the performance of the proposed protocol can further be improved by utilizing hybrid key generation protocol as proposed in \cite{hybrid}.

\section{Conclusion and Future Work Directions}
In this work, we have proposed a security mechanism for the inter-{\color{black}spacecraft} links (ISLs) first time in the literature. We have introduced a novel mobility model for future space missions, where the dynamic and static mobility components are modeled separately. By utilizing Doppler frequency shift as a secrecy source, we have proposed a first example of the mobility-based secret key generation in the literature, where this novel secrecy source paves the way for a secret key generation under static channel environments. We have provided the maximum achievable secret key rate for the Doppler frequency shift based secret key generation. The proposed  key generation procedure utilizes the nominal power spectral density samples (NSPDSs) as a projection of Doppler frequency shift measurements. A maximum likelihood (ML) estimation procedure for the NPSDSs has been introduced, and the corresponding key disagreement rate (KDR) of the procedure is analytically obtained. The numerical results highlight the consistency of the analytical bounds and approximations, and they illustrate the practical applicability of the proposed secret key generation procedure. 

{\color{black} As future work, we believe that utilizing Doppler frequency shift with oscillator-based frequency offset as a secrecy resource  would result into superposition of two independent random variables. Therefore, the total randomness obtained would increase, and we obtain higher key entropies in comparison with only carrier-frequency offset (CFO) or only Doppler frequency shift based methods. Especially when the mobility of the nodes decreases, the secret key generation system may switch from the Doppler based secrecy extraction to CFO based secret key extraction, or vice versa.} Another important future work would be comparing the performance of the proposed scheme and the fading based key generation schemes under the active physical layer attacks such as jamming. We believe that the mobility-based secret key generation would also be more resilient to active attacks since the effect of noise is negligible for the mobility-based key generation procedures.{ \color{black} Furthermore analyzing the proposed system against an attacker obtaining a radar system, and the utilization of jamming against this attacker is another possible future work.}
% The proposed mechanism ensures continuous secrecy between two distant nodes. The secrecy of the proposed method is based on the symmetric Doppler frequency measurements of the {\color{black}spacecrafts}. Theoretical expressions of the key disagreement rate (KDR) are derived considering the estimation errors at {\color{black}spacecraft}. Tight approximations to KDR expressions are obtained by using Marcum-Q functions and generalized Gauss-Laguerre quadrature (GLQ). The provided numerical results highlight the tightness of the given approximations, and indicate the applicability of the proposed key generation mechanism. As a future work, we consider a hybrid key generation mechanism that harness different secrecy resources of the physical layer (channel fading, Doppler frequency, RSS) that compensates the performance of different mechanisms.
{\color{black}  \section*{Appendix: Timing Considerations}
The timing distance between Alice and Bob is represented by $\Delta t_{AB}$, where due to the long distances for spacecrafts, the most important factor in the timing difference would be the propagation time of the transmitted data. In this case, we can say that $$\begin{aligned}
\Delta t_{AB}&= \frac{||p_A(t)-p_B(t+\Delta t_{AB})||}{c}+\frac{||p_A(t+\Delta t_{AB}+t')}{c} \\ 
&-\frac{p_B(t+2\Delta t_{AB}+t')||}{c}.\end{aligned}$$

Let us assume that $||p_A(t)-p_B(t+\Delta t_{AB})||+||p_A(t+\Delta t_{AB}+t')-p_B(t+2\Delta t_{AB}+t')||= (1+\alpha)||p_A(t)-p_B(t)||,$ and $$\Delta t_{AB}=\frac{(1+\alpha)||p_A(t)-p_B(t)||}{c}.$$ 

We assume that the velocity of the spacecrafts are almost static during the time difference between the nodes. By saying this, we need to check whether $$\Delta t_{AB}\leq \frac{||v_{max}||}{||a_A||+||a_B||},$$ where $v_{max}$ denotes the maximum velocity change tolerance, and $a_A$, $a_B$ are respectively acceleration of Alice and Bob. Considering this, we would get the following condition for the spacecrafts in order to obtain identical Doppler frequency shift measurements
\begin{equation}
\frac{(1+\alpha)||p_A(t)-p_B(t)||}{c}\leq\frac{||v_{max}||}{||a_A||+||a_B||}
\end{equation}

As a numerical example, a velocity change around 1000 m/s would only create around 6 kHz Doppler frequency shift for the 2 GHz carrier frequency. Considering the measured Doppler frequencies would be around 60 KHz for the secret key generation, we can safely neglect the velocity changes less than 1000 m/s, in this case $||v_{max}||=1000$ m/s. In order to obtain a communication link between spacecrafts, antenna design is another important consideration. In order to get signals from the Moon, the radio receivers at Earth would require a high noise tolerance due to the extremely long distances. If, we consider that the spacecrafts may can only support  the 0.1 times the distance from Earth to Moon, which is $3.8 10^8$ m. By setting $\alpha=1$, we get $||a_A||+||a_B||<5000 \text{ m/s}^2$, which is rather reasonable considering the acceleration capabilities of the current spacecrafts. 
}

\balance

\bibliographystyle{IEEEtran}
%\bibliography{references}

\input{Doppler_space_key_arxiv.bbl}

%\input{Doppler_space_key.bbl}

\end{document}